\documentclass[journal,twoside,web]{ieeecolor}
\usepackage{generic}
\usepackage{cite}
\usepackage{amsmath,amssymb,amsfonts}
\usepackage[capitalize]{cleveref}
\usepackage{mathtools}
\usepackage{bm}
\let\labelindent\relax
\usepackage[inline]{enumitem}
\usepackage{algorithmic,algorithm}
\usepackage{graphicx}
\usepackage{textcomp}
\usepackage{subcaption}
\usepackage{tikz}
\usepackage{pgfplots}
\pgfplotsset{compat=1.17}
\usepackage{adjustbox}
\usetikzlibrary{shapes,arrows}
\usepackage{eqparbox}

\tikzset{%
	block/.style    = {draw, thick, rectangle, minimum height = 3em,
		minimum width = 3em},
	sum/.style      = {draw, circle, node distance = 2cm}, 
	input/.style    = {coordinate}, 
	output/.style   = {coordinate} 
}

\newtheorem{ass}{Assumption}

\newtheorem{lem}{Lemma}
\newtheorem{thm}{Theorem}

\newtheorem{rem}{Remark}
\newtheorem{definition}{Definition}
\newtheorem{example}{Example}

\Crefname{ass}{Assumption}{Assumptions}
\Crefname{lem}{Lemma}{Lemmas}

\DeclareMathOperator*{\argmin}{arg\,min}

\newcommand{\stepupdatei}{i}
\newcommand{\stepupdate}{\step}
\newcommand{\step}{t}

\newcommand{\riskfun}{H}

\newcommand{\Tfinal}{I} 
\newcommand{\Twait}{T_{\text{wait}}}
\newcommand{\Tcollect}{T_{\text{col}}}
\newcommand{\ccollect}{c_{\text{col}}}
\newcommand{\crandom}{c_{\text{rand}}}
\newcommand{\cslack}{c_{\text{sl}}}
\newcommand{\Tn}{T_n}
\newcommand{\gammaopt}{\bm{\gamma}_{I}}
\newcommand{\marginerror}{\rho_1}
\newcommand{\marginprobability}{\rho_2}

\usepackage[acronym, style=alttree, toc=true, shortcuts, xindy, nomain, nonumberlist]{glossaries}

\newacronym{mpc}{MPC}{Model Predictive Control}
\newacronym{RMPC}{RMPC}{Robust Model Predictive Control}
\newacronym{gp}{GP}{Gaussian process}
\newacronym{SMPC}{SMPC}{Stochastic Model Predictive Control}
\newacronym{SCMPC}{SCMPC}{Scenario Model Predictive Control}
\newacronym{POMDP}{POMDP}{Partially Observable Markov Decision Process}
\newacronym{MDP}{MDP}{Markov Decision Process}
\newacronym{iid}{iid}{independent and identically distributed}
\newacronym{ocp}{OCP}{optimal control problem}
\newacronym{socp}{SOCP}{Stochastic Optimal Control Problem}
\newacronym{ISS}{ISS}{input-to-state stable}

\DeclareRobustCommand{\qed}{%
  \ifmmode 
  \else \leavevmode\unskip\penalty9999 \hbox{}\nobreak\hfill
  \fi
  \quad\hbox{\qedsymbol}}
  \newcommand{\openbox}{\leavevmode
  \hbox to.77778em{%
  \hfil\vrule
  \vbox to.675em{\hrule width.6em\vfil\hrule}%
  \vrule\hfil}}
\newcommand{\qedsymbol}{\openbox}

\usepackage{etoolbox}  
\makeatletter
\patchcmd{\algorithmic}{\addtolength{\ALC@tlm}{\leftmargin} }{\addtolength{\ALC@tlm}{\leftmargin}}{}{}
\makeatother

\newcommand{\executeiffilenewer}[3]{%
	\ifnum\pdfstrcmp{\pdffilemoddate{#1}}%
	{\pdffilemoddate{#2}}>0%
	{\immediate\write18{#3}}\fi%
}

\def\BibTeX{{\rm B\kern-.05em{\sc i\kern-.025em b}\kern-.08em
    T\kern-.1667em\lower.7ex\hbox{E}\kern-.125emX}}
\begin{document}
\title{Online Constraint Tightening in Stochastic Model Predictive Control: A Regression Approach}
\author{Alexandre Capone, \IEEEmembership{Student Member, IEEE}, Tim Brüdigam, \IEEEmembership{Member, IEEE},\\ and Sandra Hirche, \IEEEmembership{Fellow, IEEE}
\thanks{Alexandre Capone and Sandra Hirche are with the School of Computation, Information and Technology, Technical University of Munich, Arcisstraße 21, 80333 Munich, Germany. Tim Brüdigam is with the Chair of Automatic Control Engineering, Technical University of Munich, Arcisstraße 21, 80333 Munich, Germany. (e-mail: alexandre.capone@tum.de; tim.bruedigam@tum.de; hirche@tum.de.).}}

\maketitle

\begin{abstract}
Solving chance-constrained stochastic optimal control problems is a significant challenge in control. This is because no analytical solutions exist for up to a handful of special cases. A common and computationally efficient approach for tackling chance-constrained stochastic optimal control problems consists of reformulating the chance constraints as hard constraints with a constraint-tightening parameter. However, in such approaches, the choice of constraint-tightening parameter remains challenging, and guarantees can mostly be obtained assuming that the process noise distribution is known a priori. Moreover, the chance constraints are often not tightly satisfied, leading to unnecessarily high costs. This work proposes a data-driven approach for learning the constraint-tightening parameters online during control. To this end, we reformulate the choice of constraint-tightening parameter for the closed-loop as a binary regression problem. We then leverage a highly expressive \gls{gp} model for binary regression to approximate the smallest constraint-tightening parameters that satisfy the chance constraints. By tuning the algorithm parameters appropriately, we show that the resulting constraint-tightening parameters satisfy the chance constraints up to an arbitrarily small margin with high probability. Our approach yields constraint-tightening parameters that tightly satisfy the chance constraints in numerical experiments, resulting in a lower average cost than three other state-of-the-art approaches.
\end{abstract}


\section{Introduction}
\label{section:introduction}
\IEEEPARstart{S}{ystems} subject to constraint requirements are ubiquitous and pose considerable challenges for control design. A tool that has proven effective when dealing with such systems is \gls{mpc}, an optimization-based control approach that selects control inputs by minimizing a finite horizon objective function subject to input and state constraints. However, although some control applications require specific constraints to hold at all times, in practice, it is seldom possible to guarantee constraint satisfaction if pronounced uncertainties are present, and we can only satisfy constraints with a specific probability. Furthermore, in non-safety-critical applications, it is often acceptable to allow a low probability of constraint violation, e.g., in building climate control \cite{OldewurtelEtalMorari2014}, process control \cite{SchwarmNikolaou1999, JuradoEtalRubio2015}, or power systems \cite{KumarEtalZavala2018, JiangEtalDong2019}. These considerations have motivated the development of so-called \gls{SMPC} techniques, which allow constraint violations with a pre-specified probability \cite{Mesbah2016,mayne2016robust, FarinaGiulioniScattolini2016}. This approach to system constraints, commonly referred to as \textit{chance constraints}, significantly reduces conservatism compared to deterministic constraints, as possible worst-case uncertainty realizations do not necessarily need to be considered.

Although \gls{SMPC} algorithms allow a less conservative approach to system constraints compared to conventional \gls{mpc} techniques, their development is significantly more challenging. This is because the corresponding \gls{socp} generally cannot be solved exactly, and approximations must be employed. A commonly encountered approximation consists of so-called scenario approaches, which solve an optimal control problem where the cost and rate of constraint satisfaction are averaged over a finite number of sampled scenarios \cite{SchildbachEtalMorari2014,CalafioreFagiano2013,BlackmoreEtalWilliams2010}. While such methods can approximate the solution of the stochastic optimal control problem arbitrarily accurately, they often require a high number of sample scenarios, resulting in computationally expensive control algorithms. A comparatively inexpensive alternative consists of so-called analytic reformulations \cite{FarinaGiulioniScattolini2016}, which choose control inputs by solving a deterministic \gls{ocp} based on the nominal model with tightened constraints. The tightened constraints are then chosen to guarantee closed-loop satisfaction of the chance constraints. 

The most significant challenge in analytic reformulations of chance-constrained stochastic optimal control problems is the choice of constraint-tightening parameters. In particular, too loose values potentially lead to prohibitively frequent constraint violations. Conversely, excessively conservative parameter values will lead to an unnecessarily low number of constraint violations without fully exploiting the flexibility offered by chance constraints. For linear systems with zero-mean normally distributed uncertainties, we can reformulate the chance constraints as deterministic constraints for the predicted mean, which allows us to obtain constraint-tightening parameters that satisfy the chance constraints for the closed-loop system \cite{BlackmoreOnoWilliams2011, FarinaEtalScattolini2015}. Similarly, for \gls{iid} uncertainties with non-Gaussian distributions, Chebyshev inequalities can be applied to compute the worst-case magnitude of the stochastic uncertainty, which in turn allows us to compute constraint-tightening parameters that guarantee the satisfaction of the chance constraints \cite{FarinaGiulioniScattolini2016,hewing2018stochastic}. However, the resulting bound is potentially conservative, as the Chebyshev inequality is coarse for most distributions. Alternatively, sampling-based techniques can be brought to bear to compute the tightened constraints \cite{ LorenzenEtalAllgoewer2017,hewing2020scenario}. A drawback of the abovementioned approaches is that an exact description of the system and underlying uncertainty are required. Moreover, unless the distribution of the underlying uncertainty is Gaussian, they rely on conservative estimates of the worst-case error to compute constraint-tightening parameters. This potentially does not fully exploit the flexibility offered by a chance-constrained approach. This issue can be addressed by employing simulation-based methods, where Monte Carlo simulations of the stochastic system model are performed to inform the choice of constraint-tightening parameter \cite{paulson2018nonlinear,bradford2019nonlinear,bradford2020stochastic,wang2021recursive}. However, like in the non-simulation-based approaches, an exact system characterization is required to guarantee that the chance constraints are satisfied. Furthermore, full system simulations are computationally expensive and can only be performed for a finite horizon. As an alternative to simulation-based approaches, \cite{OldewurtelEtalLygeros2013} updates the constraint-tightening parameters during control. However, the method assumes a linear system with strong controllability and reachability properties often unmet in practice. Moreover, it computes the constraint-tightening parameters based on the average number of constraint violations over time, potentially exhibiting slow convergence. 

This article explores a purely data-driven approach for obtaining constraint-tightening parameters for analytic reformulations of chance-constrained stochastic optimal control problems. Our approach treats the choice of constraint-tightening parameters in a closed-loop setting as a binary regression problem, where the probability of constraint satisfaction is obtained via the stationary measure of the underlying system. The constraint-tightening parameter is then chosen by minimizing the corresponding entries subject to the chance constraints. Our approach does not require perfect knowledge of the nominal system plant nor the distribution of the underlying uncertainty, provided that the closed-loop system exhibits stable behavior. Furthermore, it is applied online, allowing potential improvements in performance during control. 


\subsection{Contribution and Structure}
\label{subsec:contribution}
The main contributions of this article are as follows. \begin{enumerate*}[label=(\roman*)]
    \item We provide sufficient conditions under which the choice of constraint-tightening parameter for the analytic reformulation of a stochastic optimal control problem can be treated as a binary regression problem.
    
    \item We present a \gls{gp} binary regression-based approach for adapting optimization parameters in an online fashion during control. To the best of our knowledge, it consists of the first binary regression-based attempt to learn chance constraints. It is considerably easier to implement than more common regression-based approaches, where full system knowledge and simulations are required.
    
    \item Our main theoretical result shows that our approach yields constraint-tightening parameters that satisfy the chance constraints up to an arbitrarily small value, provided that sufficient data is collected. We reach this result under a 
    few assumptions that are not restrictive in practice. 
    
    \item We compare our approach to different state-of-the-art approaches \cite{FarinaGiulioniScattolini2016,LorenzenEtalAllgoewer2017,hewing2018stochastic}, which more system knowledge and often yield conservative constraint-tightening parameters. 
\end{enumerate*}

The remainder of this article is structured in the following manner. We begin by formally describing the problem considered in this article in \Cref{section:problemstatement}. In \Cref{section:controlandclosedloop}, we introduce the analytic reformulation-based \gls{mpc} control law and discuss assumptions regarding the closed-loop system. \Cref{sect:gaussprocforclass} introduces \glspl{gp} for binary regression and describes how they are used to model the probability of constraint satisfaction. We show how to use \gls{gp}-based binary regression to derive an online update rule for the constraint-tightening parameters in \Cref{section:onlineupdate}. In \Cref{section:validation}, we employ a numerical simulation to illustrate the proposed method and compare it to other constraint-tightening approaches. Lastly, in \Cref{section:conclusion}, we provide some concluding remarks. The theoretical proofs of the main results can be found in Appendix - Proof of \Cref{theorem:mainresult}.

\subsection{Notation}
\label{subsec:notation}
The symbols $\mathbb{N}$, $\mathbb{R}$, $\mathbb{R}_+$, and $\mathbb{R}_-$ denote the natural, real, non-negative real, and non-positive real numbers, respectively. Plain letters denote scalars. Lowercase/uppercase boldface letters denote vectors/matrices. For a vector $\bm{a}\in\mathbb{R}^d$, we use $\Vert \bm{a}\Vert_2$ to denote its Euclidean norm. For some integer $d$, the symbol $\bm{I}_d$ denotes the $d$-dimensional identity matrix and $\bm{1}_{d_{\gamma}}$ denotes a ${d_{\gamma}}$-dimensional vector where every entry is equal to one. For a matrix $\bm{A}$, we employ $\text{eig}_{\text{max}}(\bm{A})$ and $\text{eig}_{\text{min}}(\bm{A})$ to denote the largest and smallest eigenvalue of $\bm{A}$, respectively. For an open set $\mathcal{O}\subseteq \mathbb{R}^d$, the symbol $\mathbb{I}_{\mathcal{O}}(\cdot)$ denotes the indicator function over the set $\mathcal{O}$, $\text{vol}(\mathcal{O})$ denotes the volume of $\mathcal{O}$. The symbol $\emptyset$ denotes the empty set. The symbol $p(A\vert B)$ denotes the probability density function of event $A$ conditioned on event $B$. For a set $\mathcal{S}$, we use $\vert \mathcal{S}\vert$ to denote its cardinality; for a scalar $a$, we use $\vert a\vert$ to denote its absolute value. For the states and inputs of a closed-loop system, the subscript $\tau \vert \step$ refers to the prediction at time step $\step + \tau$ given the state at time $\step$. The symbol $\mathcal{N}(\bm{\mu},\bm{\Sigma}^2)$ denotes a multivariate Gaussian distribution with mean $\bm{\mu}$ and covariance matrix $\bm{\Sigma}^2$. For scalars $a_1,\ldots,a_N$, the expression $\text{diag}(a_1,\ldots,a_N)$ corresponds to the $N$-dimensional diagonal matrix with $a_1,\ldots,a_N$ as diagonal entries.

\section{Problem Setup}
\label{section:problemstatement}

\label{subsec:dynamics}
Consider the stochastic discrete-time system

\begin{align}
\label{eq:systemeq}
    \bm{x}_{\step+1} = f(\bm{x}_{\step}, \bm{u}_{\step}, \bm{w}_{\step}), 
\end{align}
where $\bm{x}_{\step} \in \mathcal{X} \subseteq \mathbb{R}^{d_x}$ and $\bm{u}_{\step}\in \mathcal{U}\subseteq \mathbb{R}^{d_u}$ respectively denote the system's state and control input, and $ \bm{w}_{\step}\in \mathcal{W} \subseteq \mathbb{R}^{d_w}$ is \gls{iid} process noise with (potentially unknown) probability density function $p_w(\cdot)$. For the system model for $f(\cdot,\cdot,\cdot)$, we assume to have a linear approximation of the form
\begin{align}
\label{eq:linearmodel}
	f(\bm{x}_{\step}, \bm{u}_{\step}, \bm{w}_{\step}) \approx \bm{A}\bm{x}_{\step} + \bm{B}\bm{u}_{\step} + \bm{w}_{\step}.
\end{align}
We also assume that the closed-loop system is required to satisfy chance constraints. In this paper, we aim to satisfy the chance constraints in the long run, i.e., when the transient behavior has passed and the closed-loop system has reached some form of stationarity. This is expressed as
\begin{equation}
	\label{eq:chanceconstraints}
	  \lim_{T \rightarrow \infty} \mathbb{E}\left[\frac{1}{T} \sum_{t=0}^T {\mathbb{I}_{\mathbb{R}_-^{d_{\text{c}}}}\left(\bm{h}(\bm{x}_{\step})\right)} \right]\geq 1-\delta 
  \end{equation}
where $\bm{h}: \mathcal{X} \rightarrow \mathbb{R}^{d_{\text{c}}}$ is a function representing nonlinear constraints and $\delta \in (0,1)$ is a design parameter that specifies the desired risk. Unless specified otherwise, we employ $\mathbb{E}[\cdot]$ to denote the expected value with respect to the random variables $\bm{w}_{\step}$, conditioned on the initial state $\bm{x}_0$. In \Cref{section:controlandclosedloop}, we show how \eqref{eq:chanceconstraints} can be rewritten in the standard formulation for chance constraints using the concept of stationary distributions.


Our objective is then to design a control law that (approximately) minimizes the expected cost
\begin{equation}
	\label{eq:infinitehorizoncost}
	\lim \limits_{T\rightarrow\infty} {\mathbb{E}\left[\frac{1}{T}\sum\limits_{\step=0}^{T} l \left(\bm{x}_{\step},\bm{u}_{\step}\right)\right]},
\end{equation}
with stage cost $l:\mathcal{X}\times \mathcal{U} \rightarrow \mathbb{R}_{+}$, while satisfying the long-term chance constraints \eqref{eq:chanceconstraints}.

\section{Control strategy and closed-loop behavior}
\label{section:controlandclosedloop}
In this section, we describe the control law used by our approach and discuss assumptions pertaining to the closed-loop behavior of the system.

\subsection{Model predictive control and backup strategy}


 An \gls{mpc} control strategy for stochastic systems typically involves solving a finite horizon stochastic optimal control problem. However, solving a stochastic finite-horizon reformulation of \eqref{eq:chanceconstraints} is generally impossible since the linear model \eqref{eq:linearmodel} seldom corresponds to the exact system dynamics $\bm{f}$ and the probability density $p_w$ of the stochastic disturbance $\bm{w}_{\step}$ is not known perfectly. Furthermore, computing the corresponding expected cost and chance constraints is often intractable even with perfect knowledge of $\bm{f}$ and $p_w$. To address these issues, we employ a deterministic finite-time horizon reformulation of \eqref{eq:chanceconstraints} and \eqref{eq:infinitehorizoncost} with tightened constraints, which we use to find approximately optimal control inputs. This is commonly known as a finite-horizon analytic reformulation of the \gls{socp} described in \Cref{section:problemstatement} \cite{FarinaGiulioniScattolini2016}.
 In addition, we employ a relaxed version of the \gls{mpc} control strategy as a backup control law for points in the state space where the original \gls{mpc} strategy is infeasible. The backup control law is computed by greedily including slack variables in the optimal control problem until feasibility is recovered. The standard \gls{mpc} and backup strategies can then be subsumed into a single \gls{mpc}-based control law. The corresponding optimal control problem given time step $\step$ and state $\bm{x}_{ \step}$ is
 \begin{subequations}
  \label{eq:deterministicocp}
\begin{align}
	\IEEEyesnumber
	\min_{\bm{U}_{N -1\vert \step}}  \ &  \rlap{$\sum\limits_{\tau=0}^{N-1} l \left(\bm{x}_{\tau\vert \step},\bm{u}_{\tau \vert \step}\right) + l_{N}\left(\bm{x}_{N\vert \step}\right)+ \cslack\Vert \bm{s}_{i,\tau}^B \Vert_2^2$}  &\qquad \qquad   \\
	\text{s.t. } &  \rlap{$\bm{x}_{\tau+1\vert \step} = \bm{A} \bm{x}_{\tau\vert \step} + \bm{B} \bm{u}_{\tau \vert \step}$,}& \IEEEyessubnumber \\
	& \bm{u}_{\tau\vert \step} \in \mathcal{U} ,\IEEEyessubnumber\\
	& \bm{h}\left(\bm{x}_{\tau\vert \step} \right) \leq   - \bm{g}_{\stepupdatei,\tau} + \bm{s}_{\tau}^B & \\ 
 &  \bm{x}_{0\vert \step} = \bm{x}_{ \step} \\
 & \bm{0} \leq \bm{s}_{\tau}^B, \qquad \forall \ \tau\leq B \\
 & \bm{0} = \bm{s}_{\tau}^B, \qquad \forall \ \tau > B
 \IEEEyessubnumber 
\end{align}
\end{subequations}
for all $\tau \in \{0,\dots,N-1\} $,
where $\cslack>0$, $ \bm{s}_{\tau}^B$ are slack variables and $B$ is selected by solving
\[B = \min\limits_{ \tau } \tau , \quad \text{s.t. \eqref{eq:deterministicocp} is feasible}. \]
The symbol $\bm{U}_{N-1\vert \step} \coloneqq \left(\bm{u}_{0\vert \step}, \ldots, \bm{u}_{N-1\vert \step} \right)$ denotes the concatenated input sequence and 
$l_{N}:\mathcal{X} \rightarrow \mathbb{R}$ the terminal cost function. The constraint-tightening parameters 
\begin{align}
\bm{g}_{\stepupdatei,\tau}\coloneqq \bm{g}_{\stepupdate_i,\tau} \in \mathbb{R}^{d_{\text{c}}},
\end{align}
which serve to compensate for deviations between the true stochastic system \eqref{eq:systemeq} and the deterministic model used in \eqref{eq:deterministicocp}, are allowed to change at irregular time steps $\step_i$ and play a central role in our approach. In general, it is desirable to employ the lowest-possible values of $ \bm{g}_{\stepupdatei,0},\ldots,\bm{g}_{\stepupdatei,N}$ that satisfy the chance constraints. This is because low values for $ \bm{g}_{\stepupdatei,0},\ldots,\bm{g}_{\stepupdatei,N}$ imply a larger feasible region for \eqref{eq:deterministicocp} and, by extension, control inputs that potentially minimize the cost more effectively. In this article, we propose learning $ \bm{g}_{\stepupdatei,0},\ldots,\bm{g}_{\stepupdatei,N}$ from data by iteratively estimating the resulting probability of constraint satisfaction and minimizing a weighted sum of $ \bm{g}_{\stepupdatei,0},\ldots,\bm{g}_{\stepupdatei,N}$ accordingly. This is described detailedly in \Cref{section:onlineupdate}. For simplicity of exposition, we henceforth subsume the constraint-tightening parameters after the update step $\step_i$ as
\[
\left(\bm{g}_{\stepupdatei,0}^\top,\ldots,\bm{g}_{\stepupdatei,N}^\top\right)^\top \eqqcolon \bm{\gamma}_{\stepupdatei}\in \Gamma,
\]
where $\Gamma \subset \mathbb{R}^{d_{\gamma}}$ is compact and $d_{\gamma} \coloneqq N d_{\text{c}}$.

We denote the minimizer of \eqref{eq:deterministicocp} given $\bm{\gamma}_{\stepupdatei}$ as
\[\bm{U}^*_{N-1\vert\step}(\bm{x}_{\stepupdatei},\bm{\gamma}_{\stepupdatei}) \coloneqq \left(\bm{u}_{0\vert \step}^*(\bm{x}_{\step},\bm{\gamma}_{\stepupdatei}), \ldots, \bm{u}_{N-1\vert \step}^*(\bm{x}_{\step},\bm{\gamma}_{\stepupdatei}) \right).\] 
Following standard \gls{mpc}, we recursively solve \eqref{eq:deterministicocp} given the current state $\bm{x}_{\step}$ and $\bm{\gamma}_{\stepupdatei}$, then apply the first entry $\bm{u}_{0\vert \step}^*(\bm{x}_{\step},\bm{\gamma}_{\stepupdatei})$ of $\bm{U}^*_{N-1 \vert \step}(\bm{x}_{\step},\bm{\gamma}_{\stepupdatei})$ to the system. 
This implicitly specifies a control law that depends exclusively on $\bm{x}_{\step}$ and $\bm{\gamma}_{\stepupdatei}$, which we henceforth refer to as
\begin{align}
\label{eq:closedloopcontrol}
	\bm{u}^{\text{MPC}}(\bm{x}_{\step},{\bm{\gamma}}_{\stepupdatei}) = \bm{u}_{0\vert \step}^*(\bm{x}_{\step},\bm{\gamma}_{\stepupdatei}). 
\end{align}

\subsection{Closed-loop dynamics for fixed constraints}
In this paper, we do not require to know $\bm{f}$ nor $p_w$. Instead, most of the assumptions leveraged by our approach pertain to the closed-loop dynamics of the system \eqref{eq:systemeq} under the control law \eqref{eq:closedloopcontrol} and a \textit{fixed} vector of constraint-tightening parameters~$\bm{\gamma}$. To introduce and discuss our assumptions, we require the closed-loop system for a fixed $\bm{\gamma}$, given by
\begin{align}
	\label{eq:closedloop}
	\begin{split}
		\bm{x}_{\step+1} &=  \bm{f}^{\text{MPC}}_{\bm{\gamma}}\left(\bm{x}_{\step},  \bm{w}_{\step}\right) \\
		&\coloneqq \bm{f}\left(\bm{x}_{\step},\bm{u}^{\text{MPC}}(\bm{x}_{\step},{\bm{\gamma}}\right),  \bm{w}_{\step}) .
	\end{split}
\end{align}
Furthermore, for any subset $\mathcal{O}\subseteq \mathcal{X}$ and state $\bm{x}\in\mathcal{X}$, we employ 
    \[\text{Pr} \left(\bm{f}^{\text{MPC}}_{\bm{\gamma}}\left(\bm{x},  \bm{w} \right) \in \mathcal{O} \ \vert \ \bm{x}\right) = \int \limits_{\mathbb{R}^{d_x}} \mathbb{I}_{\mathcal{O}}\left(\bm{f}^{\text{MPC}}_{\bm{\gamma}}\left(\bm{x},  \bm{w} \right)\right) p_w(\bm{w}) \mathrm{d}\bm{w}
    \]
    to denote the probability of the closed-loop system \eqref{eq:closedloop} reaching the set $\mathcal{O}$ from $\bm{x}$ in a single time step.

In rough terms, our main requirements are that the closed-loop system converges to a compact set in expectation and does not exhibit periodic behavior for all $\bm{\gamma}$. This is expressed in the following assumption.

\begin{ass}
\label{ass:stableloop}
    There exists a compact subset $\mathcal{C}\subseteq \mathcal{X}$, a known measurable function $V:\mathcal{X}\rightarrow [1,\infty)$, a probability measure $\nu$ on $\mathcal{X}$, and positive constants $\beta,\tilde{\beta},K<\infty$ and $\mu<1$, such that the following conditions 
    hold for all $\bm{\gamma}\in\Gamma$:
    \begin{enumerate}[label=\roman*)]
    \vspace{0.2cm}
        \item \label{condition:minirization} $ \text{Pr} \left(\bm{f}^{\text{MPC}}_{\bm{\gamma}}\left(\bm{x},  \bm{w} \right) \in \mathcal{O} \ \vert \ \bm{x}\right) \geq \tilde{\beta}\nu\left(\mathcal{O}\right)$ holds for all $\bm{x} \in \mathcal{C}$ and all open subsets $\mathcal{O}\subseteq \mathcal{X}$.
         \vspace{0.2cm}
        \item \label{condition:lyapunov}$\mathbb{E}\left(V\left(\bm{f}^{\text{MPC}}_{\bm{\gamma}}\left(\bm{x},  \bm{w} \right)\right) \right) \leq \begin{cases} \mu V(\bm{x}), \quad & \text{if}\quad  \bm{x} \notin \mathcal{C}, \\ 
        K, \quad &\text{if}  \quad \bm{x} \in \mathcal{C}.\end{cases}$
         \vspace{0.2cm}
        \item \label{condition:aperiodicity} $\tilde{\beta} \nu(\mathcal{C})\geq \beta $.
    \vspace{0.2cm}
    \end{enumerate}
\end{ass}

Condition \ref{condition:minirization} in 
 \Cref{ass:stableloop} requires that the probability distribution of $\bm{x}$ be minorized by another probability distribution within the smaller level sets of $V$. It implies the existence of an open set of positive measure that can be reached from any point in $\mathcal{C}$ with non-zero probability.
Condition \ref{condition:lyapunov} is similar to drift conditions commonly found in stochastic control literature \cite{chatterjee2015stability}, where $V$ plays a role akin to that of a Lyapunov function. Condition \ref{condition:aperiodicity} states that a point starting in $\mathcal{C}$ reaches a subset of $\mathcal{C}$ of non-zero measure with probability larger than $0$. It implies the closed-loop system does not exhibit periodic behavior with probability one.

\Cref{ass:stableloop} can be easily shown to hold for systems with additive Gaussian noise where the noiseless realization of the system is globally asymptotically stable.
\ref{condition:minirization} and \ref{condition:aperiodicity} are not very restrictive, whereas \ref{condition:lyapunov} is similar to a stability condition. Note that $\nu$, $V$, $\beta$, $\tilde{\beta}$, $K$ and $\mu$ are assumed to be identical for all $\bm{\gamma}$, and that the function $V$ is assumed to be known. However, as will be shown later, these requirements are only necessary to derive a potentially conservative upper bound on the convergence speed of the closed-loop system to its stationary measure and are not strictly necessary in practice.

We now provide an example for a system that satisfies \Cref{ass:stableloop}.

\begin{example}

Consider the setting where the closed-loop system equations for an arbitrary fixed $\bm{\gamma}\in\Gamma$ are given by
\begin{align}
\label{eq:closedloop_example}
\bm{x}_{\step+1} = \bm{f}^{\text{MPC}}_{\bm{\gamma}}\left(\bm{x}_{\step+1},  \bm{w}_{\step} \right) \eqqcolon \tilde{\bm{f}}^{\text{MPC}}_{\bm{\gamma}}\left(\bm{x}_{\step}\right) +  \bm{w}_{\step},
\end{align}
where $\tilde{\bm{f}}^{\text{MPC}}_{\bm{\gamma}}:\mathcal{X}\rightarrow\mathcal{X}$ and $\bm{w}$ is normally distributed with zero mean and positive-definite covariance matrix $\bm{\Sigma}_{\bm{w}}^2$. In the following, we assume $\bm{\Sigma}_{\bm{w}}$ is symmetric without loss of generality. Furthermore, let
\[V(\bm{x}) = 1 + \frac{1}{2}\bm{x}^{\top} \bm{P}\bm{x},\] 
where $\bm{P}$ is a positive-definite matrix, and assume there exists a quadratic matrix $\tilde{\bm{A}}$ and positive definite matrix $\bm{M}$, such that $
\tilde{\bm{A}}^\top \bm{P}\tilde{\bm{A}} - \bm{P} + \bm{M}$
is negative-definite, $\bm{P} - \bm{M}$ is positive-definite, and
\begin{align}
\label{eq:majorizinglyapunovfunction}
(\tilde{\bm{f}}^{\text{MPC}}_{\bm{\gamma}}\left(\bm{x}\right))^\top\bm{P} (\tilde{\bm{f}}^{\text{MPC}}_{\bm{\gamma}}\left(\bm{x}\right)) \leq \bm{x}^\top \tilde{\bm{A}}^\top \bm{P}\tilde{\bm{A}}\bm{x}
\end{align}
holds for all $\bm{x}\in \mathcal{X}$ and all $\bm{\gamma} \in \Gamma$. Then, the closed-loop system \eqref{eq:closedloop_example} can be shown to satisfy \Cref{ass:stableloop} as follows. 

We begin with Condition \ref{condition:lyapunov} of \Cref{ass:stableloop}. We have
    \begin{align}
    \label{eq:inequalities_example}
    \begin{split}
        &\mathbb{E} \left[V\left(\bm{f}^{\text{MPC}}_{\bm{\gamma}}\left(\bm{x},  \bm{w} \right)\right)  \right] = \mathbb{E} \left[V\left(\tilde{\bm{f}}^{\text{MPC}}_{\bm{\gamma}}\left(\bm{x} \right)  + \bm{w}\right)\right] \\
        \leq  & \mathbb{E}\Big[1 + \frac{1}{2}\Big(\bm{x}^\top \tilde{\bm{A}}^\top \bm{P}\tilde{\bm{A}}\bm{x} +(\tilde{\bm{f}}^{\text{MPC}}_{\bm{\gamma}}\left(\bm{x}\right))^\top\bm{P}\bm{w} \\
        & \qquad + \bm{w}^\top\bm{P}\tilde{\bm{f}}^{\text{MPC}}_{\bm{\gamma}}\left(\bm{x}\right) + \bm{w}^\top\bm{P}\bm{w} \Big)\Big] \\  =  & \mathbb{E}\Big[1 + \frac{1}{2}\Big(\bm{x}^\top \tilde{\bm{A}}^\top \bm{P}\tilde{\bm{A}}\bm{x}  + \bm{w}^\top\bm{P}\bm{w} \Big)\Big]  \\ =  & 1 + \frac{1}{2}\Big(\bm{x}^\top \tilde{\bm{A}}^\top \bm{P}\tilde{\bm{A}}\bm{x}  + \text{tr}\left(\Sigma_{\bm{w}}^\top\bm{P}\Sigma_{\bm{w}}\right) \Big) \\
        <  & 1 + \frac{1}{2}\Big(  \bm{x}^\top ( \bm{P}-\bm{M})\bm{x}  + \text{tr}\left(\Sigma_{\bm{w}}^\top\bm{P}\Sigma_{\bm{w}}\right) \Big) \\
        <  & 1 + \frac{1}{2}\Big(  \text{eig}_{\text{max}}\Big(\bm{P}^{-{\frac{1}{2}}}( \bm{P}-\bm{M})\bm{P}^{-{\frac{1}{2}}}\Big)\bm{x}^\top \bm{P}\bm{x} \\&\quad  + \text{tr}\left(\Sigma_{\bm{w}}^\top\bm{P}\Sigma_{\bm{w}}\right) \Big),
        \end{split}
    \end{align}
   where the last inequality is due to 
   \[\frac{ \bm{x}^\top(\bm{P}-\bm{M})\bm{x}}{\bm{x}^\top\bm{P}\bm{x}} \leq \text{eig}_{\text{max}}\Big(\bm{P}^{-{\frac{1}{2}}}( \bm{P}-\bm{M})\bm{P}^{-{\frac{1}{2}}}\Big).\]
   Define 
   \begin{align} \label{eq:mu_example} {\mu} \coloneqq \frac{1}{2} \left(1+\text{eig}_{\text{max}}\Big(\bm{P}^{-{\frac{1}{2}}}( \bm{P}-\bm{M})\bm{P}^{-{\frac{1}{2}}}\Big)\right)\end{align} and note that 
   ${\mu} <1$. Moreover, let
    \begin{align}
    \label{eq:compact_set_example}
    \mathcal{C} = \left\{ \bm{x} \ \Big \vert \ \bm{x}^\top\bm{P}\bm{x} \leq \frac{1}{1-{\mu}}\text{tr}\left(\bm{\Sigma}_{\bm{w}}^\top\bm{P}\bm{\Sigma}_{\bm{w}}\right) +2\right\} .\end{align}
    For any $\bm{x}\notin \mathcal{C}$, from \eqref{eq:inequalities_example} we then obtain
    \begin{align*}
        &\mathbb{E} \left[V\left(\bm{f}^{\text{MPC}}_{\bm{\gamma}}\left(\bm{x},  \bm{w} \right)\right)  \right]  \\
        < &  1 + \frac{1}{2}\Big(  (2{\mu}-1)\bm{x}^\top \bm{P}\bm{x} +(1-{\mu})\bm{x}^\top\bm{P}\bm{x} -2+2{\mu}\Big) \\
        = & \mu + \frac{1}{2}\Big(  \mu\bm{x}^\top \bm{P}\bm{x} \Big) = \mu V(\bm{x}).
    \end{align*}
    Furthermore, for $\bm{x}\in \mathcal{C}$, we have
    \begin{align}
    \label{eq:K_example}
    \begin{split}
        &\mathbb{E} \left[V\left(\bm{f}^{\text{MPC}}_{\bm{\gamma}}\left(\bm{x},  \bm{w} \right)\right)  \right] 
        = 1+\frac{1}{2}\bm{x}^\top\bm{P}\bm{x} \\
        \leq & 1+  \frac{1}{2}\left(\frac{1}{1-{\mu}}\text{tr}\left(\Sigma_{\bm{w}}^\top\bm{P}\Sigma_{\bm{w}}\right) +1\right) \eqqcolon K.
        \end{split}
    \end{align}
    Hence, Condition \ref{condition:lyapunov} of \Cref{ass:stableloop} holds with $\mathcal{C}$ as in \eqref{eq:compact_set_example}, $\mu$ as in \eqref{eq:mu_example}, and $K$ as in \eqref{eq:K_example}. 

    We now show that Condition \ref{condition:minirization} holds for $\nu$ equal to the uniform distribution on $\mathcal{C}$, i.e.,
    \begin{align*}
        \nu(\mathcal{O}) = \frac{1}{\text{vol}(\mathcal{C})}\int\limits_{\mathcal{O} \cap \mathcal{C}} d\bm{x}.
    \end{align*}
    Note that, for any $\mathcal{O}$ with $\mathcal{O}\cap \mathcal{C} = \emptyset$,
    \begin{align*}
        \text{Pr} \left(\bm{f}^{\text{MPC}}_{\bm{\gamma}}\left(\bm{x},  \bm{w} \right) \in \mathcal{O} \ \vert \ \bm{x}\right) \geq \nu(\mathcal{O}) = 0
    \end{align*}
    holds trivially. Hence, we assume $\mathcal{O} \subseteq \mathcal{C}$ without loss of generality. Note that, due to \eqref{eq:majorizinglyapunovfunction} and the definition of $\mathcal{C}$,
    \begin{align*}
\Vert \tilde{\bm{f}}^{\text{MPC}}_{\bm{\gamma}}\left(\bm{x}\right)\Vert_2^2 &\leq \frac{\text{eig}_{\text{max}}(\tilde{\bm{A}}^\top\tilde{\bm{A}})}{\text{eig}_{\text{min}}(\bm{P})} \left(\frac{1}{1-{\mu}}\text{tr}\left(\Sigma_{\bm{w}}^\top\bm{P}\Sigma_{\bm{w}}\right) +2\right) \\
\Vert\bm{x}\Vert_2&\leq \frac{1}{\text{eig}_{\text{min}}(\bm{P})}\left(\frac{1}{1-{\mu}}\text{tr}\left(\bm{\Sigma}_{\bm{w}}^\top\bm{P}\bm{\Sigma}_{\bm{w}}\right) +2\right)
\end{align*}
holds for all $\bm{x}\in\mathcal{C}$. Since $p_w$ is a zero-mean multivariate Gaussian distribution with positive definite covariance matrix $\bm{\Sigma}_{\bm{w}}^2$, this implies that there exists a constant $c>0$, such that
\begin{align}
    c \leq p_w\left(\tilde{\bm{w}}-\tilde{\bm{f}}^{\text{MPC}}_{\bm{\gamma}}(\bm{x})\right)
\end{align}
holds for all $\tilde{\bm{w}}\in\mathcal{C}$ and all $\bm{x}\in\mathcal{C}$. For any $\bm{x} \in \mathcal{C}$, we then have
    \begin{align*}
        & \text{Pr} \left(\bm{f}^{\text{MPC}}_{\bm{\gamma}}\left(\bm{x},  \bm{w} \right) \in \mathcal{O} \ \big\vert \ \bm{x}\right) = \text{Pr} \left(\tilde{\bm{f}}^{\text{MPC}}_{\bm{\gamma}}\left(\bm{x}  \right)+\bm{w}  \in \mathcal{O} \ \Big\vert \ \bm{x}\right) \\ 
        =& \int \limits_{\mathcal{X}} \mathbb{I}_{\mathcal{O}}\left(\tilde{\bm{f}}^{\text{MPC}}_{\bm{\gamma}}\left(\bm{x} \right)+   \bm{w} \right) p_w(\bm{w})  \mathrm{d}\bm{w} \\ 
        =& \int \limits_{\mathcal{X}} \mathbb{I}_{\mathcal{O}}\left(   \tilde{\bm{w}} \right) p_w\left(\tilde{\bm{w}}-\tilde{\bm{f}}^{\text{MPC}}_{\bm{\gamma}}\left(\bm{x} \right) \right)  \mathrm{d}\tilde{\bm{w}}\\
        \geq & \text{vol}(\mathcal{O})\inf\limits_{\bm{x},\tilde{\bm{w}}\in\mathcal{C}} p_w\left(\tilde{\bm{w}}-\tilde{\bm{f}}^{\text{MPC}}_{\bm{\gamma}}\left(\bm{x} \right) \right) \geq c\text{vol}(\mathcal{O})  \\=& \frac{c}{\text{vol}(\mathcal{C})} \nu(\mathcal{O}),
    \end{align*}
    hence Condition \ref{condition:minirization} holds for $\tilde{\beta}=\frac{c}{\text{vol}(\mathcal{C})} $. Condition \ref{condition:aperiodicity} trivially follows with $\beta=\tilde{\beta}$.

\end{example}

\Cref{ass:stableloop} allows us to establish convergence of the distribution of the closed-loop system to a stationary probability measure. Before we state this formally, we introduce the definition of a stationary probability measure for the closed-loop system \eqref{eq:closedloop}.

\begin{definition}
\label{def:stationarymeasure}
Consider the closed-loop dynamical system \eqref{eq:closedloop} for a fixed $\bm{\gamma} \in \Gamma$. A probability measure $\pi_{\bm{\gamma}}$ on $\mathcal{X}$ is said to be stationary under $\bm{f}^{\text{MPC}}_{\bm{\gamma}}$ if, for every open subset $\mathcal{O} \subseteq \mathcal{X}$,
    \begin{align}
    \label{eq:stationarymeasure}
         \int \limits_{\mathcal{X}} \text{Pr} \left(\bm{f}^{\text{MPC}}_{\bm{\gamma}}\left(\bm{x},  \bm{w} \right) \in \mathcal{O} \ \vert \ \bm{x}\right) \pi_{\bm{\gamma}}(\mathrm{d}\bm{x}) = \pi_{\bm{\gamma}}(\mathcal{O}).
    \end{align} 
\end{definition}

Intuitively, a stationary measure is a probability measure that does not change under the system dynamics. This means that if the initial condition is distributed according to the stationary measure, then the probability distribution of the state remains constant for all future time steps. Given \Cref{ass:stableloop}, it can be shown that such a unique stationary measure exists for the closed-loop system \eqref{eq:closedloop} and that the distribution of the state converges to that of the stationary distribution. This is stated in the following.
\begin{lem}[\citenum{baxendale2005renewal}, Theorem 1.1]
\label{lem:stationarymeasure}
    Let \Cref{ass:stableloop} hold. Then, for every $\bm{\gamma} \in \Gamma$, the closed-loop system \eqref{eq:closedloop} has a unique stationary probability measure $\pi_{\bm{\gamma}}$. Furthermore, there exist positive constants $\varphi\in (0,1)$ and $\vartheta\in\mathbb{R}_+$ that are independent of $\bm{\gamma}$, such that
    \begin{align}
       \Bigg\vert \mathbb{E}\left[{\mathbb{I}_{\mathbb{R}_-^{d_{\text{c}}}}\left(\bm{h}(\bm{x}_{\step})\right)} \right] - \mathbb{E}_{\pi_{\bm{\gamma}}}\left[\mathbb{I}_{\mathbb{R}_-^{d_{\text{c}}}}\left(\bm{h}(\bm{x})\right) \right]\Bigg\vert \leq \vartheta V(\bm{x}_0) \varphi^{\step}
    \end{align}
    holds for all $\step\in \mathbb{N}$ and all $\mathbb{x}_0 \in \mathcal{X}$. Here $\bm{x}_{\step+1}$ obeys \eqref{eq:closedloop} with initial condition $\bm{x}_0$ and $\mathbb{E}_{\pi_{\bm{\gamma}}}$ corresponds to the expected value subject to $\bm{x}$ being distributed according to the stationary distribution $\pi_{\bm{\gamma}}$.
\end{lem}

\Cref{lem:stationarymeasure} states that the closed-loop system \eqref{eq:closedloop} converges exponentially to a steady state in the sense that it obeys the time-independent probability distribution $\pi_{\bm{\gamma}}$. When deriving a data-driven model of the probability of constraint satisfaction, we leverage \Cref{lem:stationarymeasure} by waiting long enough until the distribution of the state has approximately reached  $\pi_{\bm{\gamma}}$. This allows us to estimate the long-term behavior of the system from samples over a finite horizon.

Given \Cref{lem:stationarymeasure}, we can then rewrite the left-hand side of the long-term chance constraints \eqref{eq:chanceconstraints} as
 \begin{align}	
     \begin{split}	 
      &\lim_{T \rightarrow \infty} \mathbb{E}\left[\frac{1}{T} \sum_{t=0}^T {\mathbb{I}_{\mathbb{R}_-^{d_{\text{c}}}}\left(\bm{h}(\bm{x}_{\step})\right)}  
      \right] = \mathbb{E}_{\pi_{\bm{\gamma}}}\left[\mathbb{I}_{\mathbb{R}_-^{d_{\text{c}}}}\left(\bm{h}(\bm{x})\right) \right] ,
    \end{split}
  \end{align}
  where $\bm{x}_{\step+1}$ obeys \eqref{eq:closedloop}. This allows us to rewrite the chance constraints \eqref{eq:chanceconstraints} using the more common formulation 
\eqref{eq:chanceconstraints}
 \begin{align}	\label{eq:chanceconstraint_reform}
     \begin{split}	 
 	  \riskfun(\bm{\gamma})\coloneqq \text{Pr}_{ \pi_{\bm{\gamma}}}\left(\bm{h}(\bm{x})\leq \bm{0} \right) \leq 1-\delta.
    \end{split}
  \end{align}

\begin{rem}
    It is possible to employ a result similar to \Cref{lem:stationarymeasure} to reformulate the closed-loop long-term cost \eqref{eq:infinitehorizoncost} as a function of $\bm{\gamma}$, similarly to \eqref{eq:chanceconstraint_reform}. However, the resulting function is generally difficult to learn since the available measurements correspond to a stochastic process in a continuous space, indexed over $\bm{\gamma}$. A way of potentially tackling this would be to treat the measurements as noisy data, then model the expected value using, e.g., a parametric model.
\end{rem}

Our approach, discussed in \Cref{section:onlineupdate}, iteratively learns an approximation of $\riskfun$, then updates the constraint-tightening parameters by minimizing a weighted sum subject to the estimated chance constraints. 
Doing so while providing rigorous guarantees for the chance constraints is impossible without further assumptions. Hence, we make assumptions regarding $\riskfun$ and the space of constraint-tightening parameters $\Gamma$, which we explain in the following.

We assume that the space of constraint-tightening parameter $\Gamma$ is compact and contains parameters that strictly satisfy the chance constraints.

\begin{ass}
\label{ass:compactGamma}
    The space of constraint-tightening parameters $\Gamma$ is compact. Furthermore, there exists a constant $\varepsilon_{\text{feas}}>0$, such that the chance constraint are strictly satisfied with margin $\varepsilon_{\text{feas}}$, for some $\bm{\gamma}$ in $\Gamma$ i.e., \[\ \exists \bm{\gamma} \in \Gamma: \quad \riskfun(\bm{\gamma}) > 1-\delta+\varepsilon_{\text{feas}}.\]
\end{ass}
\Cref{ass:compactGamma} is not very restrictive, since $\Gamma$ can still be very large, and $\varepsilon_{\text{feas}}$ very small. It implies the problem is well posed, as the chance constraints can be satisfied strictly. In practice, choosing $\Gamma$ such that \Cref{ass:compactGamma} is satisfied can be achieved, e.g., by choosing $\Gamma$ such that very high values for $\bm{\gamma}$ are permissible, as high values typically lead to more conservative control inputs and fewer constraint violations. Alternatively, if $\Gamma$ does not satisfy \Cref{ass:compactGamma} initially, we can use our approach to estimate the maximal probability of constraint satisfaction and use this information to increase $\Gamma$.

We also make the following assumption regarding $\riskfun$.

\begin{ass}
\label{ass:rkhsnorm}
 There exists a known Lipschitz continuous sigmoid function $s: \mathbb{R}\rightarrow [0,1]$ and an unknown function $q^*: \Gamma \rightarrow \mathbb{R}$, such that the following conditions hold for all $\bm{\gamma} \in \Gamma$.
 \begin{itemize}
    \item The composition of $s$ and $q^*$ yields the long-term probability of constraint satisfaction, i.e., \[
    \riskfun(\bm{\gamma})= s\left( q^*(\bm{\gamma})\right).\]
     \item The function $q^*(\cdot)$ belongs to the reproducing kernel Hilbert space (RKHS) with reproducing kernel
     \begin{equation}
	\label{eq:kernel}
	k(\bm{\gamma},\bm{\gamma}') = \psi^{-1} k_0(\lambda \bm{\gamma},\lambda \bm{\gamma}'),
\end{equation}
where $k_0:\mathbb
{R}^{d_{\gamma}}\times \mathbb
{R}^{d_{\gamma}}\rightarrow \mathbb{R}_+$ is a nonsingular kernel, the reciprocal signal variance $\psi  \in \mathbb{R}_+$ scales the kernel, and the reciprocal lengthscale $\lambda  \in \mathbb{R}_+$ determines the bandwidth of $k(\cdot,\cdot)$. Furthermore, $k(\bm{\gamma}, \cdot)$ has
continuous partial derivatives up to order $2\alpha + 2$ for some $\alpha \in \mathbb{N}$.
 \end{itemize}
\end{ass}

For most functions $H$, we can easily find functions $q^*$ and $s$ that satisfy the first requirement of \Cref{ass:rkhsnorm}. For example, whenever $H$ assumes values strictly between $0$ and $1$, any invertible sigmoid function $s$ can be employed, e.g., the logistic function, yielding $q^*(\bm{\gamma}) = s^{-1}(H(\bm{\gamma}))$. The second requirement of \Cref{ass:rkhsnorm} restricts $q^*$ to a RKHS, specified by a kernel of the form \eqref{eq:kernel}. This assumption is not very restrictive, since the RKHSs with reproducing kernels that can be expressed as in \eqref{eq:kernel} and have derivatives up to order $2\alpha + 2$ is very rich. Examples of kernels that satisfy this requirement include the squared-exponential or Matérn kernels, which can approximate continuous functions uniformly and arbitrarily accurately in compact spaces \cite{micchelli2006universal}. Note that \Cref{ass:rkhsnorm} is considerably less restrictive than assuming a parametric structure for $q^*(\cdot)$, which would yield a significantly less expressive function space.


\section{Learning Probability of Constraint Satisfaction}
\label{sect:gaussprocforclass}

In order to choose constraint-tightening parameters $\bm{\gamma}$ that satisfy the chance constraints \eqref{eq:chanceconstraint_reform}, we aim to learn an approximation $\hat{\riskfun}_{\stepupdatei}$ of $\riskfun$ using binary regression and $n$ measurements collected during control up to time ${\stepupdate}_i$, where the number of collected data is potentially smaller than the number of time steps, i.e., $n\leq \step_i$. We achieve this using \glspl{gp}, which we introduce in this section. The idea behind learning $\riskfun$ using binary regression is shown in \Cref{fig:learning_intuition}.




\begin{figure}[t]
\centering
\includegraphics[width = 0.99\columnwidth]{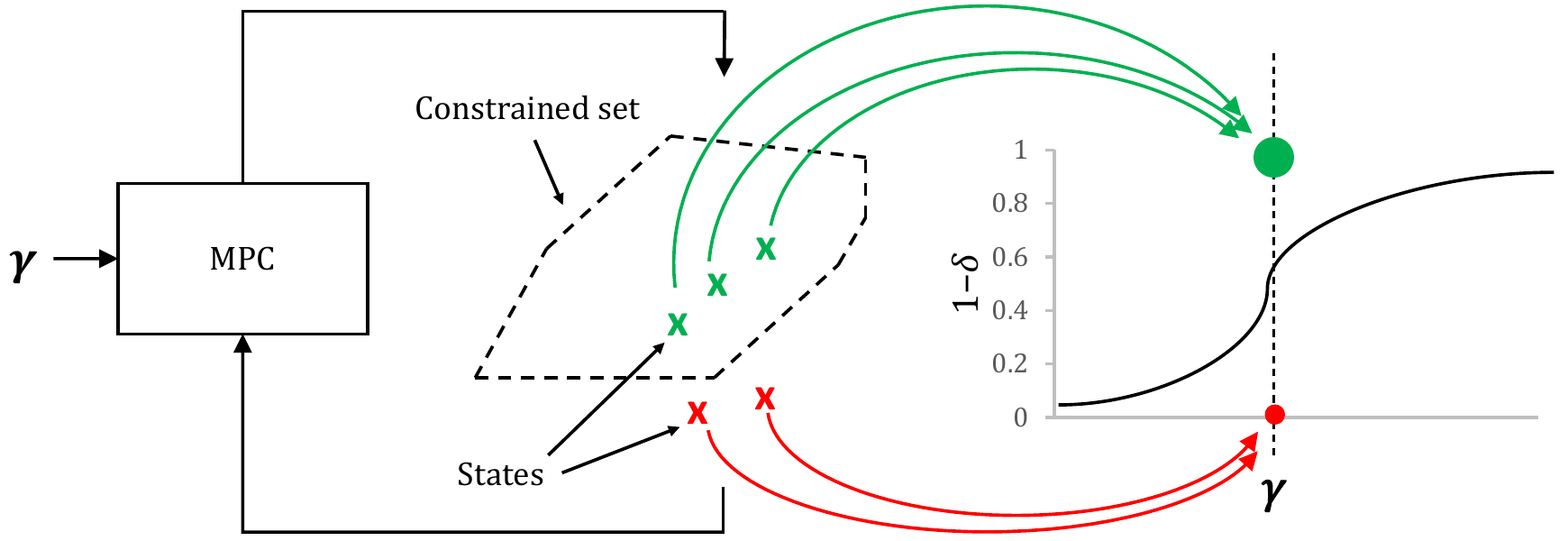}
\caption{Sketch of binary regression-based learning. The left-hand side of the figure shows the closed-loop system given $\gamma$, which results in states that are either within the constrained set (green) and outside of it (red). The goal is to learn a mapping from the constraint-tightening vector $\bm{\gamma}$ to the corresponding probability of constraint satisfaction. To this end, the number of constraint satisfactions (green) and violations (red) given $\bm{\gamma}$ are used as measurement data. In other words, $\bm{\gamma}$ is used as a training input, whereas instances where constraints are satisfied (green) and violated (red) are assigned zeroes and ones, respectively, and used as training targets. The training data is then used to fit a \gls{gp} binary regression model, shown on the right-hand side of the figure.}
\label{fig:learning_intuition}
\end{figure}

A \gls{gp} is an infinite collection of random variables, of which any finite subset is jointly normally distributed \cite{Rasmussen2006}. It is fully specified by a prior mean, which we set to zero without loss of generality, and positive-definite kernel function $k: \Gamma \times \Gamma \mapsto \mathbb{R}$, which we assume to have the same form as~\eqref{eq:kernel}. 
A \gls{gp} $q_{\text{GP}}$ is then a random variable, such that any finite number of evaluations $q_{\text{GP}}(\bm{\gamma}_1),\ldots,q_{\text{GP}}(\bm{\gamma}_n)$ corresponds to a multivariate Gaussian random variable
\begin{align}
\label{eq:normalgaussianprocess}
    \left(q_{\text{GP}}(\bm{\gamma}_1),\ldots,q_{\text{GP}}(\bm{\gamma}_n)\right)^{\top} \sim \mathcal{N}\left(\bm{0}, \bm{K}\right)
\end{align}
with covariance matrix entries $\left[\bm{K}\right]_{ij} = k(\bm{\gamma}_i, \bm{\gamma}_j)$. In order to obtain an expressive function space, we additionally assign prior probability density functions $p_{\psi}$ and $p_{\lambda}$ to $\psi$ and $\lambda$, respectively, as opposed to assuming fixed values for $\psi$ and $\lambda$. This way, we can express our belief about the smoothness and amplitude of $q_{\text{GP}}(\cdot)$ without completely ruling out other possibilities. This paper does not assume a specific form for $p_{\psi}$ and $p_{\lambda}$. However, we require that the tails satisfy the following conditions.

\begin{ass}
\label{ass:decayofsequences}
    Let $\psi_j$ and $\lambda_j$, $j\in \mathbb{N}$ be sequences of positive scalars, such that the priors on $\psi$ and $\lambda$ satisfy
    \begin{align*}
        \int\limits_{0}^{\psi_j} p_{\psi}(\psi) d\psi \leq \exp(-cj), \quad \int\limits_{\lambda_j}^{\infty} p_{\lambda}(\lambda) d\lambda \leq \exp(-cj),
    \end{align*}
    for some $c>0$. Furthermore, there exist $b_1$, $b_2\in \mathbb{R}_+$, and a sequence of positive scalars $M_j$, $j \in\mathbb{N}$, such that
    \begin{align*}
        M_j^2\psi_j\lambda_j^{-2} \geq b_1 j \quad \text{and} \quad M_j^{{ d_{\gamma}}{\alpha}^{-1}} \leq b_2 j,
    \end{align*}
    where $\alpha$ is chosen as in \Cref{ass:rkhsnorm}.
\end{ass}

\Cref{ass:decayofsequences} is not restrictive, as the priors $p_{\psi}$ and $p_{\lambda}$ are purely design choices. It dictates how fast $p_{\psi}$ and $p_{\lambda}$ have to decay relative to each other. \Cref{ass:decayofsequences} is satisfied, e.g., by priors with exponentially decaying tails \cite{ghosal2006posterior}. In rough terms, \Cref{ass:decayofsequences} limits how fast the posterior can change as new measurements are collected, which is useful for establishing convergence of the learned model. 


The idea behind \gls{gp}-based binary regression is to squash the \gls{gp} through a sigmoid function $s: \mathbb{R} \rightarrow [0,1]$, obtaining the prior probability associated with the value of a binary random variable
\begin{align}
\label{eq:sigmoidprior}
    p( y=1 \vert \bm{\gamma}) = s(q_{\text{GP}}(\bm{\gamma})).
\end{align}
This way, \gls{gp} evaluations with a high positive value are assigned a high probability and vice versa. 
Given training data \[\mathcal{D}_{\stepupdatei} = \left\{\bm{\gamma}_j, y_j \right\}_{j=1,\ldots,n} \eqqcolon\{\bm{G},\bm{y}\}\] 
with binary training labels $y_j \in \{0,1\}$ and fixed hyperparameters $\psi$, $\lambda$, the \gls{gp} binary regression model is then conditioned on $\mathcal{D}_{\stepupdatei}$, yielding the predictive model
\begin{align}
    p(y = 1 \vert \mathcal{D}_{\stepupdatei}, \bm{\gamma}) = \int \limits_{\mathbb{R}} s(q_{\text{GP}})p(q_{\text{GP}}\vert \mathcal{D}_{\stepupdatei}, \bm{\gamma}) dq_{\text{GP}},
\end{align}
where $\bm{\gamma}$ denotes the test input. The posterior distribution of the latent variable $q_{\text{GP}}(\bm{\gamma})$ given $\bm{\gamma}$ is computed as \cite{Rasmussen2006}
\begin{align}
\label{eq:prediction}
    p(q_{\text{GP}}\vert \mathcal{D}_{\stepupdatei}, \bm{\gamma}) = & \int \limits_{\mathbb{R}^n} p(q_{\text{GP}} \vert \bm{\gamma}) p(\bm{q}_{\text{GP}} \vert \mathcal{D}_{\stepupdatei}) d\bm{q}_{\text{GP}},
    \end{align}
    where $\bm{q}_{\text{GP}}\coloneqq \left(q_{\text{GP}}(\bm{\gamma}_1),\ldots,q_{\text{GP}}(\bm{\gamma}_\step)\right)^{\top}$, the posterior of the latent variable values is given by
    \begin{align*}
    p(\bm{q}_{\text{GP}} \vert \mathcal{D}_{\stepupdatei}) = &  \frac{p(\bm{y}\vert\bm{q}_{\text{GP}})p(\bm{q}_{\text{GP}}\vert \bm{G})}{p(\bm{y}\vert \bm{G})},
\end{align*}
and the terms $p(\bm{y}\vert\bm{q})$, $p(\bm{q}\vert \bm{G})$, $p(\bm{y}\vert \bm{G})$ are in turn computed using \eqref{eq:normalgaussianprocess} and \eqref{eq:sigmoidprior}. To choose the hyperparameters $\psi$ and $\lambda$, we can sample from the corresponding posterior given the data
\begin{align}
	\label{eq:posterior}
	p(\psi,\lambda \vert \mathcal{D}_{\stepupdatei}) = \frac{p(\bm{y} \vert \bm{G},\psi,\lambda) p(\psi) p(\lambda)}{p(\bm{y}\vert \bm{G})}.
\end{align}
Alternatively, the hyperparameters $\psi$ and $\lambda$ can be chosen by maximizing the posterior \eqref{eq:posterior}.

\begin{rem}
In practice, \eqref{eq:prediction} typically does not have an analytical solution, and we have to resort to numerical approximations, e.g., Laplace's approximation or Markov Chain Monte Carlo. For the latter approach, convergence guarantees can be obtained; see, e.g., \cite{haggstrom2002finite}. 
\end{rem}

\begin{figure*}
\centering
\begin{adjustbox}{width=0.9\textwidth}
\begin{tikzpicture}[auto, thick, node distance=2cm, >=triangle 45]
	\draw
		node at (6,-1)[draw, thick, rectangle,minimum height = 1.6cm, name=gpblock,text width=3.5cm,align=center]{\centering {Train GP model} 

		{\hspace{-0.5cm}\begin{tikzpicture}
			\begin{axis}[
				axis on top = true,
				axis x line = bottom,
				axis y line = left,
				grid = none,
				ticks=none,
				width=4cm,
				height=2.3cm,
				xlabel shift=-0.4cm,
				ylabel shift=-0.4cm,
				xlabel =  $\bm{\gamma}$,
				ylabel =  $\hat{\riskfun}_{{\stepupdate_i}}(\bm{\gamma})$
				]
			\addplot[
				blue,
				domain = -8:8,
				samples = 100
				]
				{sin(x/0.01)/12 + 1/(1+exp(-x))};
			\end{axis}
		\end{tikzpicture}
	\vspace{-0.5cm}
	}
	}
	;
	\draw		node at (13,-1)[draw, thick, rectangle,minimum height = 2.5cm, name=constraintblock,text width=6.5cm,align=center]{\centering Update data 
		
		$\mathcal{D}_i = \mathcal{D}_i \cup \left\{\bm\gamma_{\stepupdatei},y_{\step}\right\}$
  
    \vspace{2pt}
		$y_{\step} = \mathbb{I}_{\mathbb{R}_+} \Big( \bm{h}(\bm{x}_{\step}) \Big)$

  $ {\stepupdate_i}+ \Twait(\bm{x}_{\stepupdate_i}) \leq \step < {\stepupdate_i}+ \Twait(\bm{x}_{\stepupdate_i})+ \Tcollect$
	}
	;
	\draw
	node at (0,-1)[draw, thick, rectangle,minimum height = 2cm, name=tighteningblock,text width=4cm,align=center]{Tighten constraints
 \vspace{4pt}
	\begin{tabular}[t]{rl}
	$\min\limits_{\bm{\gamma}}$& $\bm{a}^{\top} \bm{\gamma},$\\
	$\text{s.t.}$& $\hat{\riskfun}_{\stepupdatei}(\bm{\gamma}) \geq  1 - \delta$
\end{tabular}
	}
	;
	\draw
	node at  (4,-6.5) [draw, thick, rectangle,minimum height = 3cm, name=mpcblock,text width=6.5cm,align=center] { Model predictive control law
		
		{\tiny

\begin{align*}
	\IEEEyesnumber
	\min_{\bm{U}_{N -1\vert \step}}  \ &  \rlap{$\sum\limits_{\tau=0}^{N-1} l \left(\bm{x}_{\tau\vert \step},\bm{u}_{\tau \vert \step}\right) + V_{\text{f}}\left(\bm{x}_{N\vert \step}\right)+ \Big\Vert \bm{s}_{i,\tau}^B \Big\Vert_2^2$}  &\qquad \qquad   \\
	\text{s.t. } &  \rlap{$\bm{x}_{\tau+1\vert \step} = \bm{A} \bm{x}_{\tau\vert \step} + \bm{B} \bm{u}_{\tau \vert \step}$,}& \IEEEyessubnumber \\
	& \bm{u}_{\tau\vert \step} \in \mathcal{U} ,\IEEEyessubnumber\\
	& \bm{h}\left(\bm{x}_{\tau\vert \step} \right) \leq   - {\color{red}\bm{g}_{\stepupdatei,\tau}} + \bm{s}_{\tau}^B & \\ 
 &  \bm{x}_{0\vert \step} = \bm{x}_{ \step} \\
 & \bm{0} \leq \bm{s}_{\tau}^B, \qquad \forall \ \tau\leq B \\
 & \bm{0} = \bm{s}_{\tau}^B, \qquad \forall \ \tau > B
 \IEEEyessubnumber 
\end{align*}
  
 }}
	node at (12,-6.5) [draw, thick, rectangle,minimum height = 2.5cm,text width=3.5cm,align=center] (plantblock) {Partially Unknown Plant

	{$\bm{x}_{\step+1} = f(\bm{x}_{\step}, \bm{u}_{\step}, \bm{w}_{\step})$}}
	;
	\draw
	node at (15.5,-6.55)(bullet1){\Large\textbullet}
	node at (0,-3.07)(bullet2){\textcolor{red}{\Large\textbullet}}
	;
	\draw[->] (mpcblock) -- node {$\bm{u}_{\step}$} (plantblock);
	;
 \draw[-] (plantblock) -- node {$\bm{x}_{\step}$}(15.5,-6.5)(bullet1);
	\draw[->] (bullet1.south) -- (15.5,-7) -- (15.5, -9) -| node {}(mpcblock);
  \draw[->] (bullet1.north)-- (15.5,-7) -- (constraintblock.south -| bullet1);
	\draw[->] (constraintblock.west) -- node {$\mathcal{D}_{\stepupdatei}$}(gpblock.east);
	\draw[->,color=red] (tighteningblock.south) --node {$\bm{\gamma}_{\stepupdatei}$}(0,-3)--(0,-6.5) -- (mpcblock.west);
	\draw[->,color=red] (tighteningblock.south) --(0,-2)--(0,-3) -- (12.5,-3)-|(constraintblock.south);
	\draw[->] (gpblock.west) -- node {$\hat{\riskfun}_{\stepupdatei}$}(tighteningblock.east);
	\draw [color=black,dashed](-2.5,-3.8) rectangle (17,1);
	\node at (-1.7,0.25) [above=5mm, right=0mm] {\textsc{Binary Regression-Based Constraint Tightening}};
	\draw [color=black,dashed](-2.5,-10) rectangle (17,-4.3);
	\node at (-1.7,-9.2) [below=5mm, right=0mm] {\textsc{MPC control law and plant}};
\end{tikzpicture}
\end{adjustbox}
	\caption{Closed-loop control with constraint-tightening.}
\label{fig:closedlooptightening}
\end{figure*}
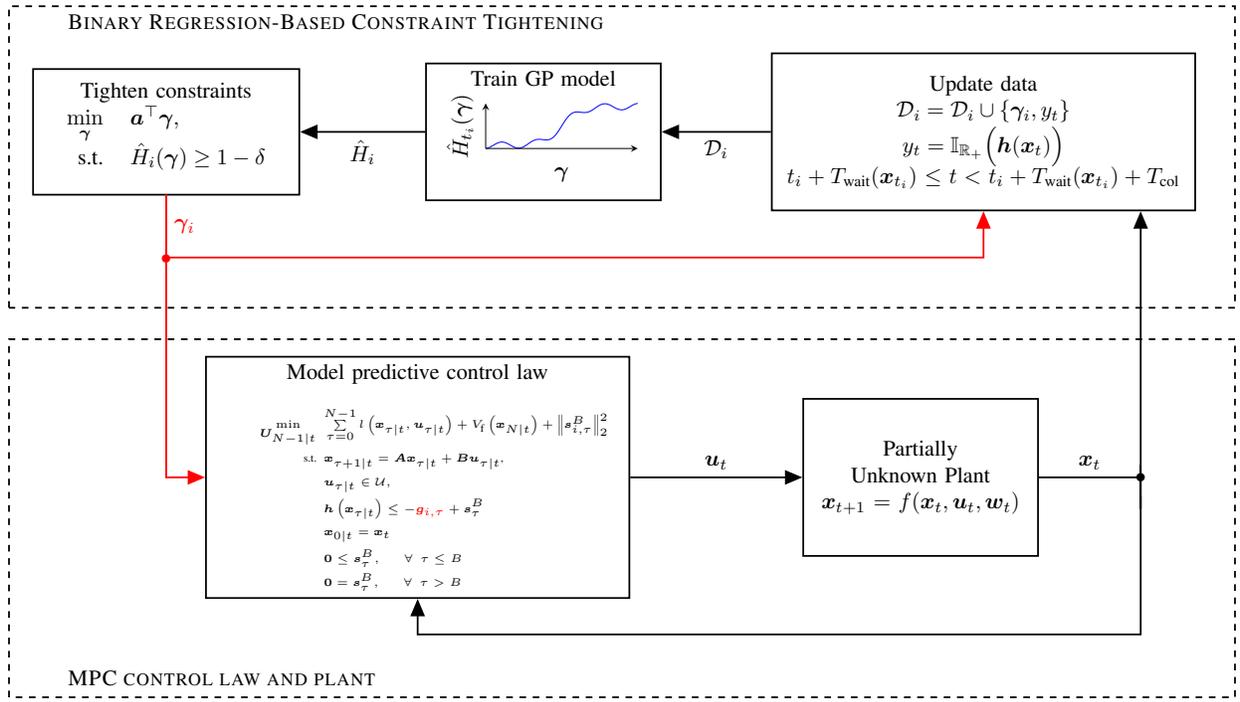

\section{Online Update of Constraint Parameters}
\label{section:onlineupdate}

We now present our approach, followed by corresponding theoretical guarantees. The goal of our approach is to find a solution to the optimization problem
\begin{align}
	\label{eq:gammasearch}
		\min\limits_{\bm{\gamma} \in \Gamma} \  {\bm{a}}^{\top} \bm{\gamma}, \qquad 
		\text{s.t.} \  {\riskfun}(\bm{\gamma}) \geq 1 - \delta,
\end{align}
where ${\bm{a}} \in \mathbb{R}_+$ is a vector of positive weights. This choice of cost function is motivated by the fact that if the entries of $\bm{\gamma}$ are small, the feasible region of \eqref{eq:deterministicocp} is large, leading to potentially more effective control inputs. Our method applies the \gls{mpc} control law \eqref{eq:closedloopcontrol} using $\bm{\gamma}_{i}$, which in turn is updated at time steps $\step_i$, specified in the following. To this end, our approach employs a \gls{gp}-based approximation $\hat{\riskfun}_i$ of $\riskfun$, trained using collected measurements. This is repeated until $\Tfinal$ data points have been collected, after which a final optimization step is performed to determine the approximately optimal constraint-tightening parameters. 

Given the state $\bm{x}_{\step}$ and a vector $\bm{\gamma}_{\stepupdatei}$ of constraint-tightening parameters, our method applies the \gls{mpc} control law \eqref{eq:closedloopcontrol} using $\bm{\gamma}_{\stepupdatei}$. 
To guarantee that the collected data is representative of the long-term closed-loop behavior, $\bm{\gamma}_{\stepupdatei}$ is kept constant for $\Twait(\bm{x}_{\stepupdate_i})$ time steps, where $\bm{x}_{\stepupdate_i}$ corresponds to the state when $\bm{\gamma}_{\stepupdatei}$ was last updated, and $ \Twait :\mathcal{X} \rightarrow \mathbb{R}_+ $ is a function to be specified in the following. The reason for waiting $\Twait(\bm{x}_{\stepupdate_i})$ times steps is that it ensures that the distribution of the state is sufficiently close to that of the stationary distribution $\pi_{\bm{\gamma}_{\stepupdatei}}$ when we start collecting data. After waiting $\Twait(\bm{x}_{\stepupdate_i})$ steps, our approach then collects  $\Tcollect$ measurements before updating $\bm{\gamma}_{\stepupdatei}$. The role of $\Tcollect$ is threefold. First, it ensures that we obtain sufficient information about the distribution of $\mathbb{I}_{\mathbb{R}_-^{d_{\text{c}}}}\left(\bm{h}(\bm{x}_{\step})\right)$ for each training input. Secondly, it limits the frequency with which $\bm{\gamma}_{\stepupdatei}$ is updated, which can be leveraged to reduce the overall computational load. Lastly, it allows us to keep collecting data without restarting the procedure, which would require us to wait until the new stationary distribution given the updated vector $\bm{\gamma}_{\stepupdatei}$ has been approximately reached. After $\Tcollect$ measurements have been collected, the data collected up until the current time step is used to learn a \gls{gp}-based approximation  $\hat{\riskfun}_{{i}}$ of $\riskfun$, as described in \Cref{sect:gaussprocforclass}. The constraint-tightening parameters $\bm{\gamma}_{\stepupdatei+1}$ are then updated by approximately solving the optimization problem
\begin{align}
	\label{eq:gammasearchgp}
		\bm{\gamma}_{{i+1}}=\argmin\limits_{\bm{\gamma} \in \Gamma} \  {\bm{a}}^{\top} \bm{\gamma}, \qquad 
		\text{s.t.} \  \hat{\riskfun}_{{i}}(\bm{\gamma}) \geq 1 - \delta,
\end{align}
where ${\bm{a}} \in \mathbb{R}_+$ is a vector of positive weights. This choice of cost function is motivated by the fact that if $\bm{\gamma}_i$ is small, the feasible region of \eqref{eq:deterministicocp} is large, leading to potentially more effective control inputs. In addition to the optimization-based update step \eqref{eq:gammasearchgp}, we also include a random update step that takes place every $\crandom$ updates, where $\crandom \in \mathbb{N}\backslash \{0\}$ is arbitrary but fixed, i.e., whenever
\[\step = {\stepupdate_{i}}, \quad i=n\crandom, \quad n\in \mathbb{N},\]
and whenever \eqref{eq:gammasearchgp} is infeasible. The role of random updates is to ensure sufficient coverage of the constraint-tightening parameter space $\Gamma$. This procedure is repeated until $\Tfinal$ data points have been collected, after which we choose the final vector of constraint-tightening parameters from all previously collected inputs
\begin{align}
	\label{eq:gammasearchgp_final}
		\gammaopt =\argmin\limits_{\bm{\gamma} \in \left\{\bm{\gamma}_1, \ldots, \bm{\gamma}_{\Tfinal} \right\}} \  {\bm{a}}^{\top} \bm{\gamma}, \qquad 
		\text{s.t.} \  \hat{\riskfun}_{{\Tfinal}}(\bm{\gamma}) \geq 1 - \delta.
\end{align}
We require this final step because the class of functions that satisfies \Cref{ass:rkhsnorm} is very general: our model $\hat{\riskfun}_i$ converges to the true function $\riskfun$ over all collected data points $\bm{\gamma}_1, \ldots, \bm{\gamma}_i$, but we cannot generally exclude pathological cases where this does not imply uniform convergence, e.g., if the Lipschitz constant of $\hat{\riskfun}_i$ grows unbounded. Although this step poses no practical issues, it is not required, e.g., if the Lipschitz constant of $\hat{\riskfun}_i$ is bounded for all $i$ or if $\Gamma$ is finite. 
These steps are summarized in \Cref{alg:onlineupdatealgorithm}. An illustration without the random update step and final optimization step \eqref{eq:gammasearchgp_final} is given in \Cref{fig:closedlooptightening}.

\begin{algorithm}[t]
\caption{SMPC with Online Constraint Parameter Update}
\label{alg:onlineupdatealgorithm}
\begin{algorithmic}[1]
 \renewcommand{\algorithmicrequire}{\textbf{Input:}}
 \renewcommand{\algorithmicensure}{\textbf{Output:}}
 \REQUIRE{Initial parameter $\bm{\gamma}_0$, \gls{gp} model $\hat{\riskfun}_0$, desired risk $\delta$, search space $\Gamma$ }
 \ENSURE{Constraint-tightening parameters $\gammaopt$}
  \STATE Initialize data set $\mathcal{D}_0 = \{\}$
  \STATE Set ${\stepupdate_i}=0$, set $\bm{x}_{{\stepupdate_i}}=\bm{x}_{0}$
  \WHILE{$i< \Tfinal$}
  \STATE Apply $\bm{u}^{\text{MPC}}(\bm{x}_{\step},{\bm{\gamma}_{\stepupdatei}})$ to system using \eqref{eq:closedloopcontrol}, measure $\bm{x}_{\step+1}$
    \STATE Update time step $\step=\step+1$ \vspace{-0.1cm}
  \IF{$ {\stepupdate_i}+ \Twait(\bm{x}_{\stepupdate_i}) \leq \step < {\stepupdate_i}+ \Twait(\bm{x}_{\stepupdate_i})+ \Tcollect$}
  \vspace{0.1cm}
   \STATE Collect data and update data set \[\mathcal{D}_{\stepupdatei} = \mathcal{D}_{\stepupdatei} \cup \left\{\bm{\gamma}_{\stepupdatei}, \mathbb{I}_{\mathbb{R}_-} \left( \bm{h}(\bm{x}_{\step}) \right)\right\}\] 
  \ELSIF{$ \step = {\stepupdate_i}+ \Twait (\bm{x}_{\stepupdatei}) + \Tcollect$}
  \vspace{0.1cm}
  \STATE Update $i=i+1$, $\stepupdate_i=\step$, $\bm{x}_{{\stepupdate_i}}=\bm{x}_{{\step}}$
  \vspace{0.1cm}
  \STATE Use $\mathcal{D}_{\stepupdatei}$ to train \gls{gp} model $\hat{\riskfun}_{{\stepupdatei}}$
  \IF{$i=n\crandom, \quad n\in \mathbb{N}$ or \eqref{eq:gammasearchgp} is infeasible}
  \STATE Sample $\bm{\gamma}_{\stepupdatei}$ from uniform distribution on $\Gamma$
  \ELSE
   \STATE Update $\bm{\gamma}_{\stepupdatei}$ by solving \eqref{eq:gammasearchgp}
  \ENDIF
  \ENDIF
  \ENDWHILE
  \STATE Choose $\gammaopt$ by solving \eqref{eq:gammasearchgp_final}
 \end{algorithmic} 
 \end{algorithm}

\begin{rem}
    The method proposed in this paper can be directly extended to the episodic setting, where a finite-horizon OCP problem is addressed, provided that the initial state distribution is \gls{iid}. In this case, we can employ \Cref{alg:onlineupdatealgorithm} by setting $\Twait(\bm{x})=0$ for all $\bm{x}\in \mathcal{X}$ and $\Tcollect$ equal to the episode horizon. The only difference is that the state $\bm{x}_{\step}$ is reset every $\Tcollect$ steps according to the initial state distribution. The equivalent of the stationary distribution $\pi_{\bm{\gamma}}$ is then the joint distribution of the states within the entire horizon $\Tcollect$.
\end{rem}

We now show that if we choose $\Tcollect$, $\Tfinal$ and the function $\Twait$ high enough, the approximately optimal parameter $\gammaopt$ satisfies the chance constraints \eqref{eq:chanceconstraints} up to an arbitrarily small margin. This corresponds to our main theoretical result and is stated in the following.

\begin{thm}
\label{theorem:mainresult}
Let \Cref{ass:stableloop,ass:rkhsnorm,ass:compactGamma,ass:decayofsequences}
hold. For any $\Tfinal$, choose $\Twait$ and $\Tcollect$ such that
    \begin{align}
    \label{eq:stepslowerbound}
        \Twait(\bm{x}) &\geq \frac{1}{-\log(\varphi)}\left( \log\left({\vartheta}V(\bm{x})\right) + \Tfinal \log\left(2\right) \right),  \\
    \label{eq:collectionstepslowerbound}
    \Tcollect &\geq \ccollect \Tfinal,
    \end{align}
where $\varphi$ and $\vartheta$ are as in \Cref{lem:stationarymeasure}, and $ \ccollect>0$ is an arbitrary but fixed parameter. Then, for any $\marginerror>0$ and $\marginprobability \in (0,1)$, there exists a $\Tfinal\in \mathbb{N}$, such that, with probability at least $1-\marginprobability$, \eqref{eq:gammasearchgp_final} is feasible and the corresponding solution satisfies
\begin{align}
{\riskfun}(\gammaopt)  \geq 1-\delta -\marginerror.
\end{align}
\end{thm}
\begin{proof}
    See Appendix - Proof of \Cref{theorem:mainresult}.
\end{proof}

\Cref{theorem:mainresult} states that if the number of collected data points $\Tfinal$ is high enough, then, with high probability $1-\marginprobability$, the final optimization step in \Cref{alg:onlineupdatealgorithm} is feasible, and we obtain a vector of constraint-tightening parameters $\gammaopt$ that satisfies the chance constraints \eqref{eq:chanceconstraints} 
up to an arbitrarily small margin $\marginerror$. However, to achieve this, we must choose the function that specifies the waiting time $\Twait$ and $\Tcollect$ and such that \eqref{eq:stepslowerbound} and \eqref{eq:collectionstepslowerbound} are satisfied, respectively. The bounds \eqref{eq:stepslowerbound} 
 and \eqref{eq:collectionstepslowerbound} grow linearly with $\Tfinal$. In the case of $\Tcollect$, this is because our model $\hat{\riskfun}_{\Tfinal}$ converges to the true function $\riskfun$ on average, hence ensuring that $\Tcollect$ corresponds to a fixed fraction $\ccollect$ of $\Tfinal$ allows us to recover pointwise bounds. The bound on $\Twait$ \eqref{eq:stepslowerbound} is directly tied to \Cref{lem:stationarymeasure}, which states that the data becomes more representative of the stationary measure $\pi_{\bm{\gamma}_i}$ of the closed-loop system as time increases. It implies that we can extract more information from a higher amount of data $\Twait$ if the corruption in the data is sufficiently low, which is an intuitive result.

\begin{rem}
    In practice, determining the parameters $\varphi$ and $\vartheta$ required to compute the bound \eqref{eq:stepslowerbound} for $\Twait$ can be difficult. Furthermore, the resulting bound can be conservative, stipulating a long waiting time before we can start collecting data. However, the bound \eqref{eq:stepslowerbound} corresponds only to a sufficient condition, not a necessary one, and in practice, a lower value can be potentially picked. In some cases, we can determine convergence to the stationary measure $\pi_{\bm{\gamma}_i}$ empirically, e.g., by checking whether the state has converged to a neighborhood of a specific point.
\end{rem}

\begin{rem}
    The bound \eqref{eq:collectionstepslowerbound}, the random update step, and the final optimization step in \Cref{alg:onlineupdatealgorithm} are all related to the nonparametric nature of our the \gls{gp}-based model. In the case of strictly parametric models, stronger convergence results can be obtained \cite{shen2001rates}, simplifying the requirements of \Cref{theorem:mainresult}. However, such a model is generally more difficult to justify than the nonparametric model used in this paper, as the corresponding function space is typically considerably less expressive.
\end{rem}

\begin{rem}
    Although we do not provide a result that establishes convergence of $\gammaopt$ to the optimum of \eqref{eq:gammasearch}, the same techniques employed in this paper can be employed to derive such a result, e.g., using Berger's maximum theorem \cite{aliprantis2006infinite}. However, to this end, we require the additional assumption that the subset of $\Gamma$ that satisfies $\riskfun(\bm{\gamma})<1-\delta$ is convex in $\delta$. In our experiments, we observed that $\riskfun$ is locally monotonically increasing, which is potentially sufficient. Furthermore, this property caused $\gammaopt$ to converge to the minimizer of \eqref{eq:gammasearch} in all experiments.
\end{rem}

\subsection{Discussion}
\label{section:discussion}

A clear strength of the proposed approach is that we can obtain guarantees while taking the backup strategy into account, 
which is of high practical relevance. It is also worth noting that other forms of backup controller are also allowed, provided that the closed-loop system satisfies \Cref{ass:stableloop}. 
Our approach 
could potentially also be used to inform such a choice in practice, e.g., by applying and comparing backup control strategies other than the slack variable-based approach presented here.



Our approach adapts the constraint-tightening parameter $\bm{\gamma}$ by exploring the parameter space $\Gamma$, where the observed constraint violations inform the exploration. Although our approach is guaranteed to converge a $\bm{\gamma}$ that satisfies the chance constraints due to \Cref{theorem:mainresult}, the corresponding algorithm might choose values of $\bm{\gamma}$ that will momentarily lead to a high rate of constraint violations. In cases where this is undesirable, it can be mitigated by adopting a safe strategy, e.g., by adaptively increasing the space of constraint tightening parameters $\Gamma$ such that initially only conservative values are permissible, then gradually allowing less conservative values.

One of the more notable challenges when employing \glspl{gp} is their poor scalability with the number of data points, as training scales cubically with the amount of data and evaluation scales quadratically \cite{Rasmussen2006}. This can be addressed using different frequently encountered tools, e.g., sparse GPs \cite{snelson2005sparse} or networks of local experts \cite{deisenroth2015distributed}.

If the \gls{mpc} horizon $N$ or the number of constraints $d_c$ are very high, then the dimension $d_{\gamma}$ of the constraint-tightening space $\Gamma$ is potentially large, making searching for an optimizer of \eqref{eq:gammasearchgp} difficult. However, this can be easily mitigated by setting 
\[\Gamma = \left\{ \bm{D} \tilde{\bm{\gamma}} \ \big \vert \ \tilde{\bm{\gamma}} \in \tilde{\Gamma} \right\}, \]
where $\bm{D}$ is a matrix and $\tilde{\Gamma}$ is a low-dimensional search space, then conditioning the \gls{gp} model on training inputs from $\tilde{\Gamma} $. This is exemplified in \Cref{section:validation}.

A further aspect of the proposed algorithm that should be considered is the initial guess for  $\bm{\gamma}_0$. Though \Cref{theorem:mainresult} requires a portion of the parameter space $\Gamma$ to be explored thoroughly, convergence can potentially be sped up significantly if a good initial vector of constraint-tightening parameters $\bm{\gamma}_0$ is provided, as this will decrease the size of the explored region significantly compared to a poor initial guess. 

\section{Numerical Example}
\label{section:validation}

In this section, we aim to illustrate \Cref{theorem:mainresult} and assess the performance of the proposed approach using a numerical experiment. Moreover, we compare our approach to different state-of-the-art approaches.


\subsection{Simulation Setup}
\label{subsection:simulationsetup}

We consider the discrete-time system corresponding to the linearized model of a DC-DC converter \cite{CannonEtalCheng2011,LorenzenEtalAllgoewer2017}, given by
\begin{equation}
\label{eq:simenvironment}
    \bm{x}_{\step+1} = \begin{bmatrix} 1 & 0.0075 \\ -0.143 & 0.996\end{bmatrix} \bm{x}_{\step} + \begin{bmatrix} 4.798 \\ 0.115\end{bmatrix} u_t + \bm{w}_{\step},
\end{equation}
where we consider an initial state of  $\bm{x}_0=(0,0)^\top$, chance constraints \[\Pr\left( [1 \ \ 0]\bm{x}_{\step}\leq 0 \right) \geq 1-\delta, \quad  \forall \ t \in \mathbb{N}, \]
and a varying risk parameter $\delta \in [0.6, 0.01]$. The cost to be minimized is given by \eqref{eq:infinitehorizoncost}, with immediate cost
\begin{align*}
    l(\bm{x},{u}) = \bm{x}^\top \bm{Q}  \bm{x} + u R u,
\end{align*}
and weights \[\bm{Q} = \begin{pmatrix}1& 0 \\ 0 & 10\end{pmatrix}, \quad {R}~=~1.\] We assume the control input to be bounded as $\vert u_t \vert \leq 0.2$ and consider different distributions for the stochastic disturbance $\bm{w}_{\step}$, to be specified in the following.

For the \gls{mpc} algorithms, we employ an \gls{mpc} horizon time of $N =  10$ and assume to know the state matrix $\bm{A}$ and the input matrix $\bm{B}$ in \eqref{eq:simenvironment}. 
We employ a quadratic terminal cost function of $l_{N}(\bm{x})=\bm{x}^\top\bm{P} \bm{x}$, where $\bm{P}$ satisfies the Lyapunov function $ \bm{P} = \bm{A} \bm{
P}\bm{A}^\top  + \bm{Q}$. For the backup strategy, the cost penalties are obtained by multiplying the slack variables with $\cslack=10^16$.

\subsection{Illustration of \Cref{theorem:mainresult}}
\label{subsection:illustrationoftheorem1}

\begin{figure}[t]
\centering
\includegraphics[width = 0.99\columnwidth]{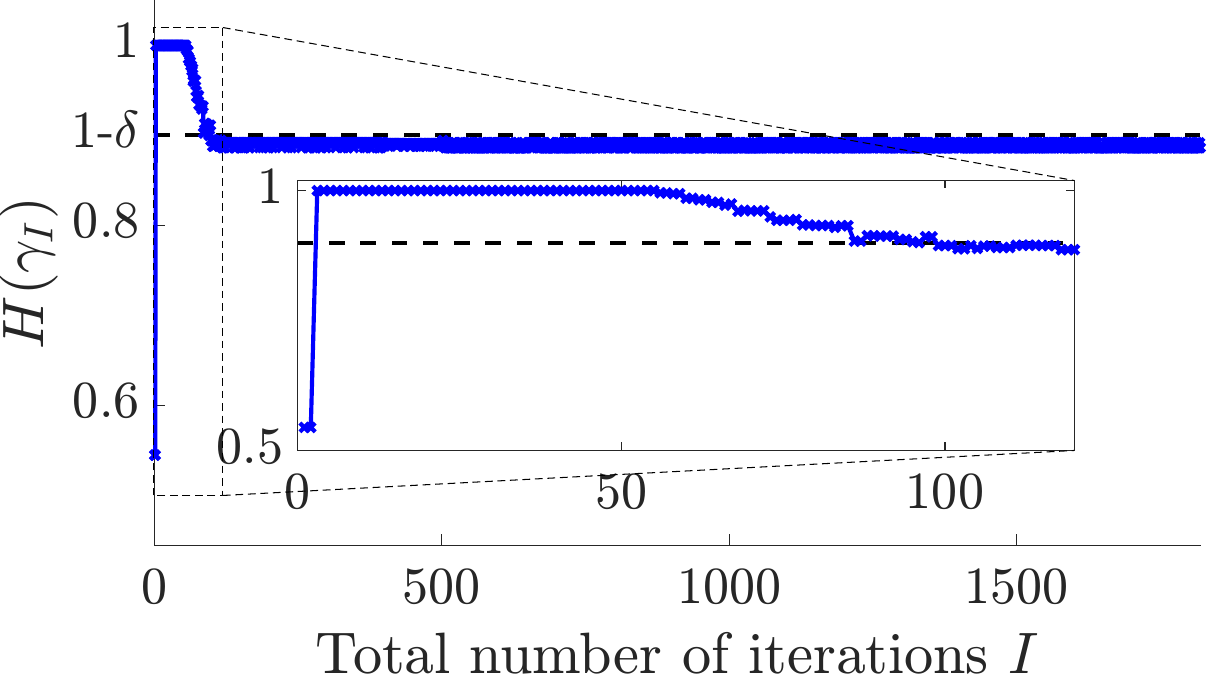}
\caption{Long-term probability of constraint satisfaction $\riskfun(\bm{\gamma}_{\Tfinal})$ obtained with parameters $\bm{\gamma}_{\Tfinal}$ returned by \Cref{alg:onlineupdatealgorithm} for different total number of iterations $\Tfinal$.} 
\label{fig:progress_of_HgammaI}
\end{figure}

\begin{figure}
\centering
\includegraphics[width = 0.99\columnwidth]{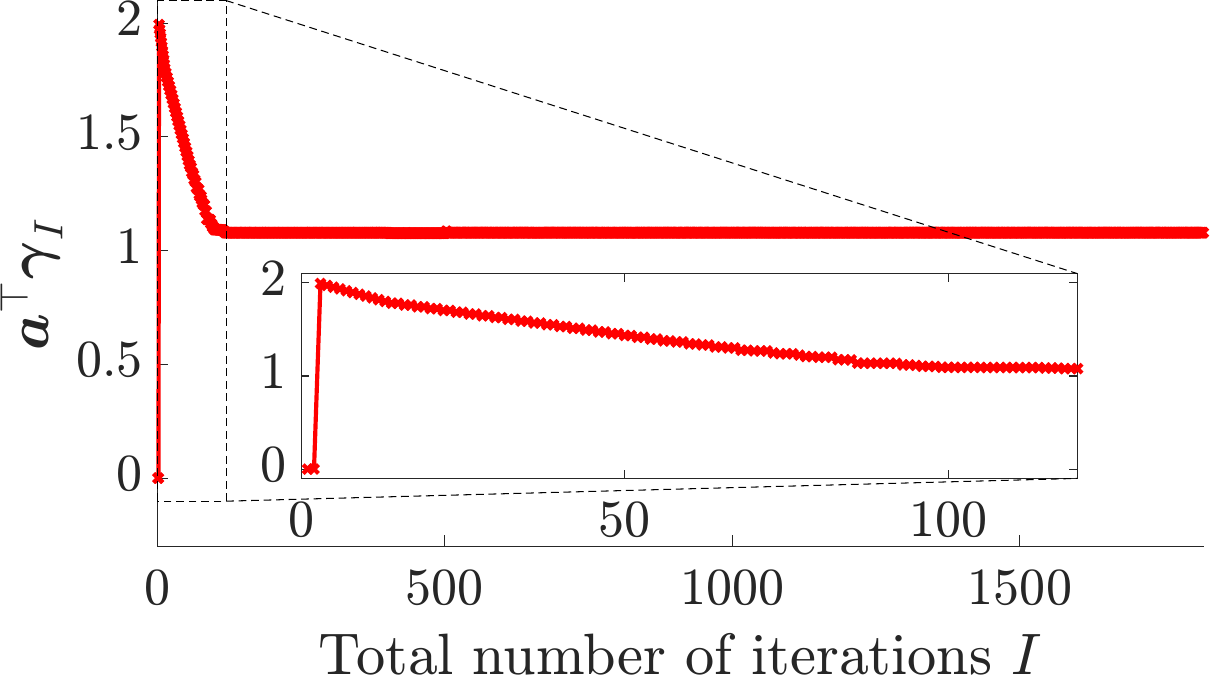}
\caption{Weighted sum $\bm{a}^\top\bm{\gamma}_{\Tfinal}$ of entries of constraint-tightening parameters $\bm{\gamma}_{\Tfinal}$ returned by \Cref{alg:onlineupdatealgorithm} for different total number of iterations $\Tfinal$.} 
\label{fig:progress_of_aTgammaI}
\end{figure}

\begin{figure}[t]
\centering
\includegraphics[width = 0.99\columnwidth]{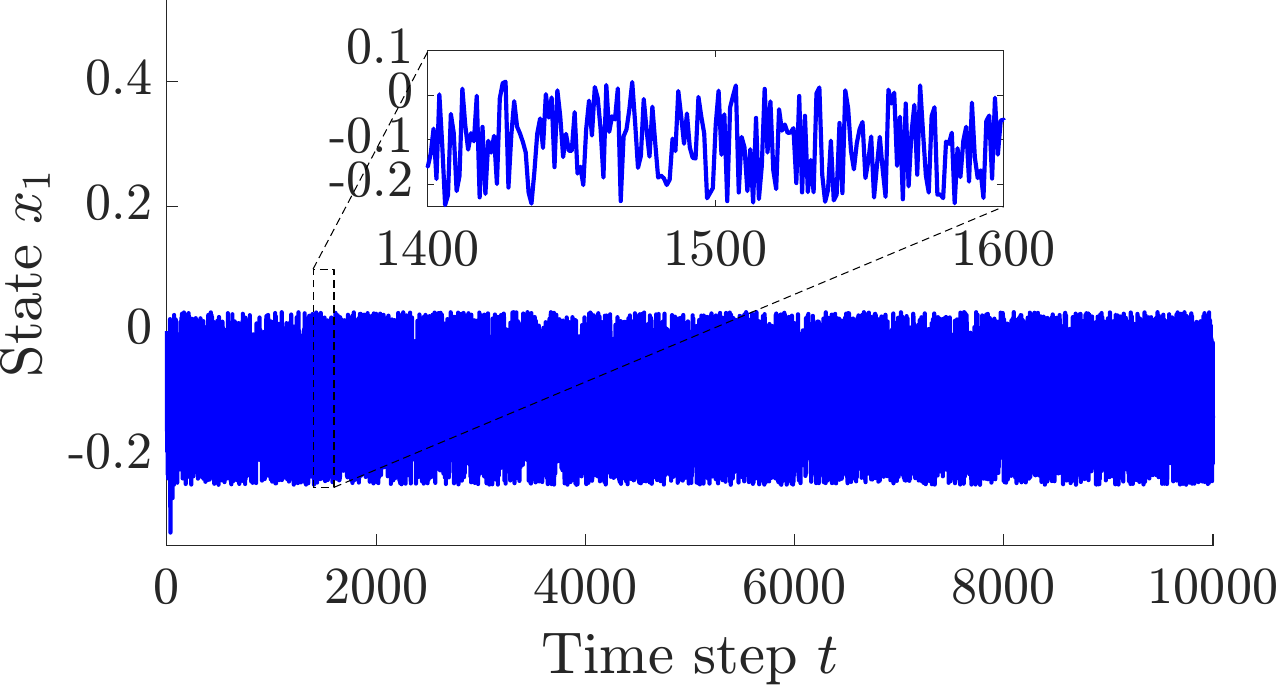}
\caption{State $x_1$ obtained for a simulation with $10^4$ time steps and the constraint-tightening parameters $\bm{\gamma}_{\Tfinal}$ returned by \Cref{alg:onlineupdatealgorithm} after $\Tfinal=1800$ iterations. The empirical probability of constraint satisfaction is $0.8912< 1 -\delta - \marginerror$ for $\marginerror = 0.01$, which is to be expected from \Cref{theorem:mainresult}.} 
\label{fig:x1_final_gammaI}
\end{figure}

To illustrate \Cref{theorem:mainresult}, we run \Cref{alg:onlineupdatealgorithm} multiple times for varying $\Tfinal=1,..,1800$ and a fixed risk of $\delta=0.1$. We consider process noise $\bm{w}_{\step}$ that is sampled from a uniform distribution on $[-0.14,0.14]$. To model $\riskfun$, we employ a squared-exponential kernel 
  \begin{equation}
	\label{eq:sekernel}
	k(\bm{\gamma},\bm{\gamma}') = \psi^{-1} \exp\left(-\frac{\lambda^2}{2}\Vert\bm{\gamma} - \bm{\gamma}'\Vert_2^2\right),
\end{equation}
with Gaussian priors for the log hyperparameters  
$\log({\psi}) \sim \mathcal{N}(-1,1)$, and
a modified Gaussian error function
\begin{align}
    s(z) = 1 + \frac{1}{\sqrt{\pi}} \int\limits_0^{z} e^{-\xi^2} d \xi
\end{align} 
as sigmoid function. To keep computations cheap while handling multiple thousands of training data points, we employ a sparse \gls{gp} approximation with all different training inputs as inducing inputs \cite{snelson2005sparse,rasmussen2010gaussian}.
We consider a one-dimensional space of constraint-tightening parameters 
\[\Gamma = \left\{   \tilde{\gamma}  \bm{1}_{d_{\gamma}}  \ \big \vert \ \tilde{{\gamma}} \in [-1,0.2] \right\}. \]
This choice of parameter space allows us to search over a one-dimensional space instead of a ${d_{\gamma}}$-dimensional one. Note~$\Gamma$ admits $\bm{\gamma}$ with negative entries, allowing us to avoid unnecessary conservatism, e.g., if the unconstrained \gls{mpc} algorithm already satisfies the chance constraints. For the optimization \eqref{eq:gammasearchgp}, we employ a weight vector of $\bm{a}=\bm{1}_{d_{\gamma}}$. To solve \eqref{eq:gammasearchgp}, we perform an exhaustive search of $\Gamma$ starting at the lowest end until we find a feasible point. This is inexpensive since \eqref{eq:gammasearchgp} corresponds to a one-dimensional optimization problem with a monotonically increasing cost. We set the initial constraint-tightening parameters to $\bm{\gamma}_0 = \bm{0}$, which does not satisfy the chance constraints. The values for $\Twait(\bm{x})$ stipulated by \Cref{theorem:mainresult} can be conservative. Instead of computing them directly, we choose $\Twait(\bm{x})=500$ for all $\Tfinal$. This choice is motivated by the observation that the state $x_1$ converges to a seemingly stationary distribution after only a handful of steps for all starting values close to $\bm{0}$. We additionally set $\Tcollect = 5000$ for all $\Tfinal$ and $\crandom = 100$, i.e., every hundredth iteration $i$ corresponds to a random search step. 

\begin{figure*}
\centering
\begin{subfigure}[b]{0.49\textwidth}
\includegraphics[width = 0.99\columnwidth]{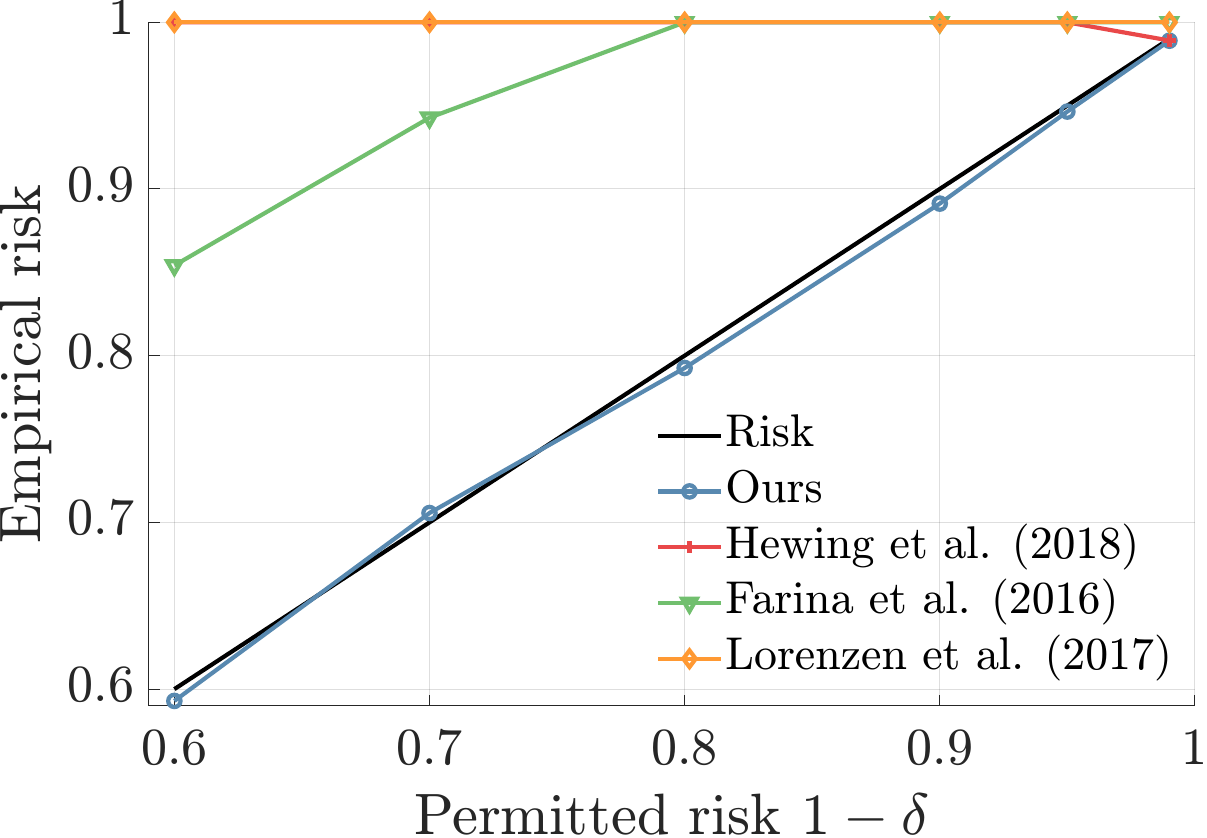}
\caption{Empirical rate of constraint satisfaction for single simulation with $2 \times 10^5$ time steps.} 
\label{fig:constr_satisf_uniform}
\end{subfigure}
\hfill
\begin{subfigure}[b]{0.49\textwidth}
\includegraphics[width = 0.99\columnwidth]{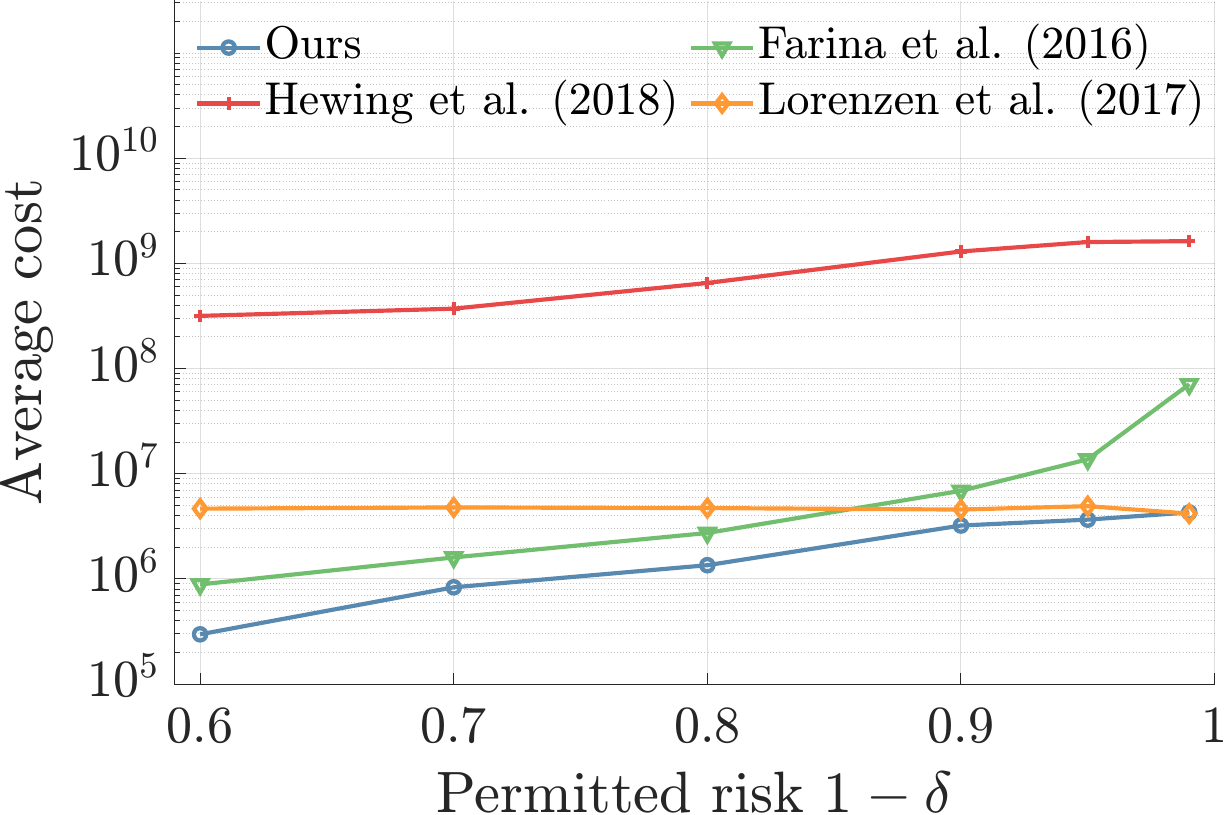}
\caption{Average cost for single simulation with $2 \times 10^5$ time steps.} 
\label{fig:cost_uniform}
\end{subfigure}
\caption{Empirical rate of constraint satisfaction and average cost for uniformly distributed uncertainties on $[-0.14,0.14]$. }
\label{fig:uniform_uncertainties_results}
\end{figure*}

\begin{figure*}
\centering
\begin{subfigure}[b]{0.49\textwidth}
\includegraphics[width = 0.99\columnwidth]{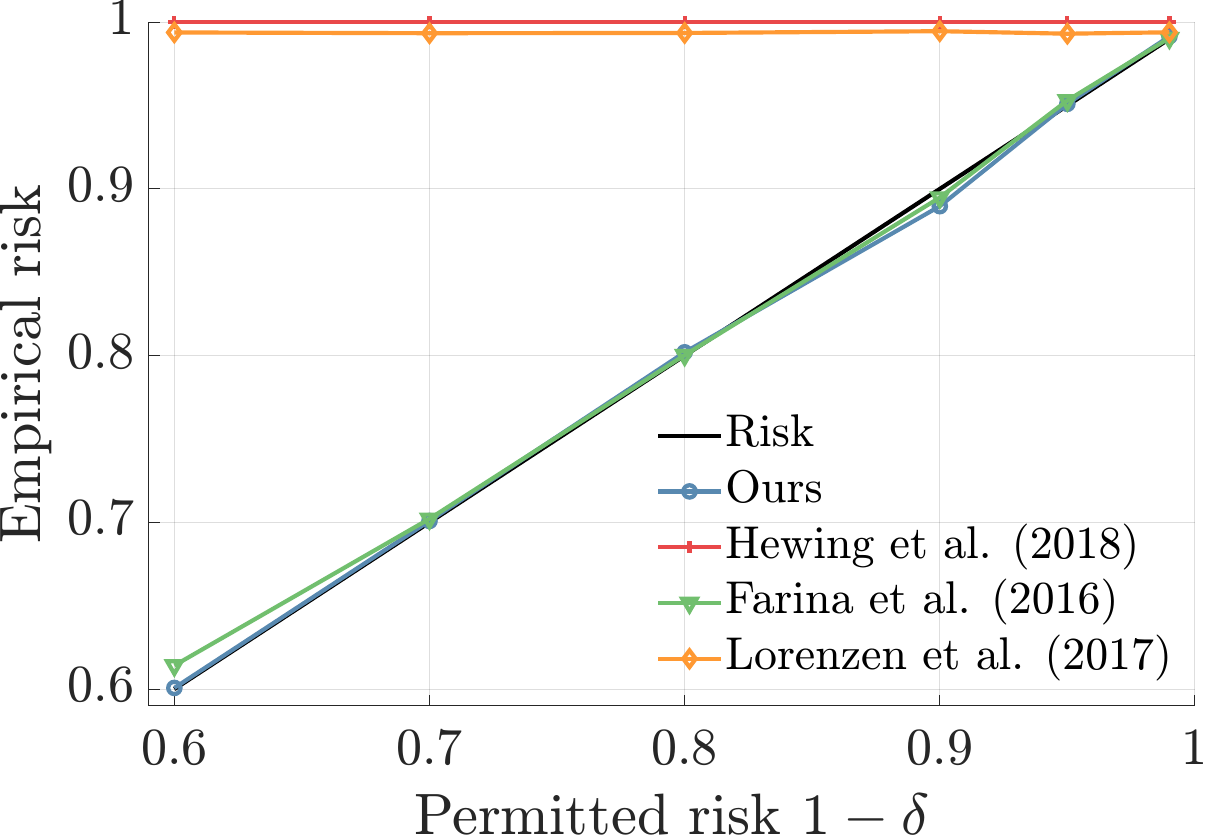}
\caption{Empirical rate of constraint satisfaction for single simulation with $2 \times 10^5$ time steps.} 
\label{fig:constr_satisf_gaussian}
\end{subfigure}
\hfill
\begin{subfigure}[b]{0.49\textwidth}
\includegraphics[width = 0.99\columnwidth]{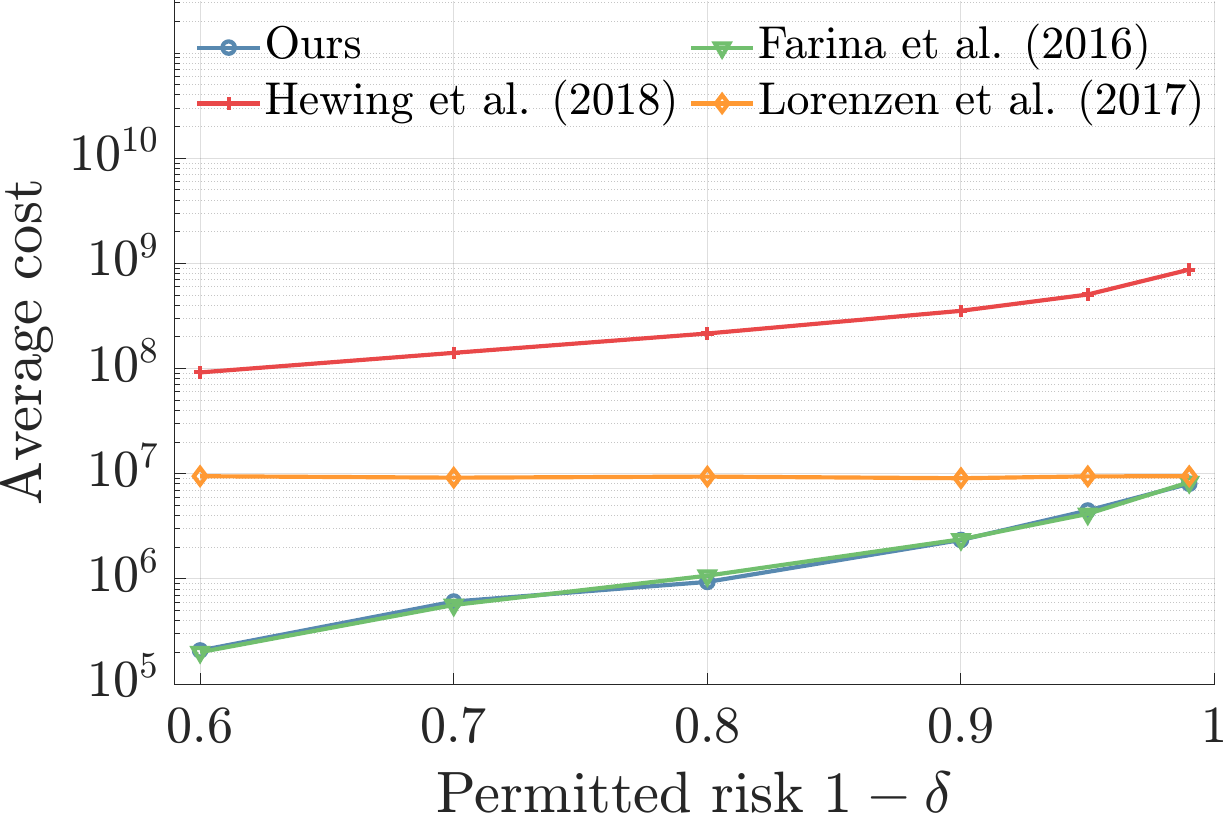}
\caption{Average cost for single simulation with $2 \times 10^5$ time steps.} 
\label{fig:cost_gaussian}
\end{subfigure}
\caption{Empirical rate of constraint satisfaction and average cost for zero-mean Gaussian uncertainties on with variance $0.08$. }
\label{fig:gaussian_uncertainties_results}
\end{figure*}

For every $\gammaopt$, we approximate $\riskfun(\bm{\gamma}_{\Tfinal})$ by measuring the average rate of constraint satisfaction over a simulation of $5000$ steps. The progress of $\riskfun(\bm{\gamma}_{\Tfinal})$ with $\Tfinal$ can be seen in \Cref{fig:progress_of_HgammaI}, that of the cost function $\bm{a}^\top\bm{\gamma}_{\Tfinal}$ can be seen in \Cref{fig:progress_of_aTgammaI}. Since $\bm{\gamma}_0 = \bm{0}$ does not satisfy the chance constraints, none of the data initially used to train the \gls{gp} model $\hat{\riskfun}$ corresponds to a feasible point. This causes the optimization problem \eqref{eq:gammasearchgp} to be infeasible. Hence, the constraint-tightening parameters $\bm{\gamma}_i$ generated by \Cref{alg:onlineupdatealgorithm} are sampled from a uniform distribution on $\Gamma$, which in turn is reflected in $\bm{\gamma}_{\Tfinal}$ if $\Tfinal$ is chosen small. However, after a handful of random searches, a vector $\bm{\gamma}_{\Tfinal}$ of constraint-tightening parameters is found that strictly satisfies the chance constraints. From there, the constraint-tightening parameters $\bm{\gamma}_{\Tfinal}$ returned by \Cref{alg:onlineupdatealgorithm} gradually decrease until the chance constraints are approximately satisfied with equality, after which $\bm{\gamma}_{\Tfinal}$ stays approximately constant. This is because $\hat{\riskfun}$ approximates ${\riskfun}$ accurately if enough data is available. If we choose $\Tfinal>100$, then the final vector of constraint-tightening parameters $\bm{\gamma}_{\Tfinal}$ will have converged to a value that satisfies the chance constraints up to a margin of $\marginerror<0.01$, which is to be expected by \Cref{theorem:mainresult}. The state $x_1$ for a sample simulation with the constraint-tightening parameters $\bm{\gamma}_{\Tfinal}$ obtained with \Cref{alg:onlineupdatealgorithm} and $\Tfinal=1800$ is shown in \Cref{fig:x1_final_gammaI}.

\subsection{Comparison with Existing Approaches}
\label{subsection:comparison}

We now compare our approach to \cite{FarinaGiulioniScattolini2016,LorenzenEtalAllgoewer2017,hewing2018stochastic}, three different state-of-the-art approaches for computing constraint-tightening parameters of analytic reformulations of \gls{SMPC} problems. The approach presented in \cite{FarinaGiulioniScattolini2016} computes constraint-tightening parameters by determining a $1-\delta$ credible interval for the process noise distribution and adjusting the constraints accordingly over the prediction horizon. 
Similarly, the approach of \cite{hewing2018stochastic,hewing2020recursively,hewing2020scenario} leverages the process noise distribution to compute probabilistic reachable sets, guaranteeing that the chance constraints are satisfied in a closed-loop fashion. The method of \cite{LorenzenEtalAllgoewer2017} employs a sampling-based approach to compute the constraint-tightening parameters. Furthermore, it includes a stochastic constraint on the first step to ensure the recursive feasibility of the optimization problem. We note that we do not employ a gain matrix, as suggested by \cite{LorenzenEtalAllgoewer2017,hewing2018stochastic}. However, their approaches are also applicable without a gain matrix, provided that the system matrix $\bm{A}$ is stable, which holds in this setting. Our method employs the same setup as in \Cref{subsection:illustrationoftheorem1}, except for the total number of iterations, which we set to $\Tfinal=150$. We consider the same input constraints and backup control policy for all approaches. 

We consider two different settings, one where the entries of the process noise are uniformly distributed on $[-0.14,0.14]$, and one where they are normally distributed with mean zero and standard deviation $0.08$. We compute constraint-tightening parameters for seven different risk parameters 
\[
\delta \in \{0.6, 0.7, 0.8, 0.9, 0.95, 0.99\}.
\] 
We simulate each setting once over $20000$ time steps. We additionally compare the resulting average cost for each approach.

The results for the uniformly distributed process noise entries can be seen in \Cref{fig:uniform_uncertainties_results} and those for the Gaussian distributed process noise entries in \Cref{fig:gaussian_uncertainties_results}. Our approach satisfies the chance constraints for all $\delta$ up to a small margin of $\marginerror<0.01$, which is to be expected by \Cref{theorem:mainresult}. Furthermore, it satisfies the constraints \textit{tightly} up to a margin of $0.01$, which indicates that $\gammaopt$ converged to a local minimum of \eqref{eq:gammasearchgp}. 
This is also reflected in the average cost, which is lower than that of the other approaches. This difference is particularly accentuated for the methods of \cite{FarinaGiulioniScattolini2016} and \cite{hewing2018stochastic} in the setting with uniformly distributed uncertainties due to their use of Chebyshev inequalities, which are conservative. In the case of Gaussian distributed uncertainties, the approach of \cite{FarinaGiulioniScattolini2016} performs similarly to our approach, which is to be expected since the corresponding constraint-tightening parameters tightly satisfy the chance constraints.

\section{Conclusion}
\label{section:conclusion}
We have presented a binary regression approach for choosing the constraint-tightening parameters of the analytic reformulation of a stochastic optimal control problem. We have provided sufficient conditions under which a finite number of system samples can approximate the long-term rate of constraint satisfaction. Our algorithm leverages these conditions to compute constraint-tightening parameters that provably satisfy the chance constraints up to an arbitrarily low error using a \gls{gp} binary regression model. In numerical simulations, our approach yielded constraint-tightening parameters that tightly satisfied chance constraints, allowing it to outperform state-of-the-art approaches in average cost. Potential future research directions include obtaining a convergence rate for the proposed method and experimenting with different binary regression tools, e.g., logistical or probit regression approaches.

\section*{Appendix - Proof of \Cref{theorem:mainresult}}

To show that the constraint-tightening parameters returned by our algorithm converge to values that respect the chance constraints, we first need to show that the \gls{gp}-based binary regression model is increasingly accurate as more data is employed to train it. To this end, we employ the following result, which applies to arbitrary data sets.
\begin{lem}
\label{lem:ghosal}
    Let \Cref{ass:compactGamma,ass:rkhsnorm,ass:stableloop,ass:decayofsequences} hold, and consider a sequence of data sets, increasing in size, where the training inputs are fixed but potentially different between sets, i.e.,
    \[\tilde{\mathcal{D}}_j = \{\tilde{\bm{\gamma}}_{j,n}, \tilde{y}_{j,n} \}_{n=1,\ldots,j},\qquad j\in\mathbb{N}\]
    where
    \[\tilde{\bm{\gamma}}_{j,1},...,\tilde{\bm{\gamma}}_{j,j} \in \Gamma\]
    and $\tilde{y}_{j,1},...,\tilde{y}_{j,j}$ are Bernoulli random variables with
    \[\tilde{y}_{j,n} = \mathbb{I}_{\mathbb{R}_-^{d_{\text{c}}}}\left(\bm{h}(\bm{x})\right), \quad \bm{x}\sim \pi_{\tilde{\bm{\gamma}}_{j,n}}.\]
    Here $\pi_{\tilde{\bm{\gamma}}_{j,n}}$ denotes the unique stationary measure for the closed-loop system \eqref{eq:closedloop} using $\tilde{\bm{\gamma}}_{j,n}$. Let $\tilde{\riskfun}_{j}$ denote the posterior \gls{gp} conditioned on $\tilde{\mathcal{D}}_j$. Then, for every $\marginerror,\marginprobability > 0$, there exists a $J\in\mathbb{N}$, such that 
    \begin{align*}
        \Pr\left(\frac{1}{j} \sum \limits_{n=1}^j \vert \riskfun(\tilde{\bm{\gamma}}_{j,n}) - \tilde{\riskfun}_{{j}}(\tilde{\bm{\gamma}}_{j,n})\vert > \marginerror \ \Big\vert \ \tilde{\mathcal{D}}_j\right) < \marginprobability.
    \end{align*}
    holds for all $j\geq J$. 
    \vspace{0.1cm}
\end{lem}
\begin{proof}
    The proof follows straightforwardly from \Cref{lem:stationarymeasure} and \cite[Theorem 2]{ghosal2006posterior}.
\end{proof}

\Cref{lem:ghosal} implies that if the measurements $\tilde{y}_{j,1},\ldots, \tilde{y}_{j,j}$ are sampled from the stationary distributions $\pi_{\tilde{\bm{\gamma}}_j,n}$ for arbitrary but fixed training inputs $\tilde{\bm{\gamma}}_{j,1},\ldots, \tilde{\bm{\gamma}}_{j,j} \in \Gamma$,
then the accuracy of the \gls{gp} model grows in expectation with the data size $j$. A direct consequence is that we obtain convergence as the data size grows, independently of the training input locations.

\begin{lem}
\label{lem:ghosal_extended}
    Let \Cref{ass:compactGamma,ass:rkhsnorm,ass:stableloop,ass:decayofsequences} hold. Then, for every $\marginerror,\marginprobability > 0$, there exists a $J\in\mathbb{N}$, such that 
     \begin{align*}
        \Pr\left(\frac{1}{j} \sum \limits_{n=1}^j \vert \riskfun(\tilde{\bm{\gamma}}_{j,n}) - \tilde{\riskfun}_{{j}}(\tilde{\bm{\gamma}}_{j,n})\vert > \marginerror \ \Big\vert \ \tilde{\mathcal{D}}_j\right) < \marginprobability.
    \end{align*}
    holds for all $j\geq J$ and all data sets
    \[\tilde{\mathcal{D}}_j = \{\tilde{\bm{\gamma}}_{j,n}, \tilde{y}_{j,n} \}_{n=1,\ldots,j},\qquad j\in\mathbb{N}\]
    with training inputs
    \[\tilde{\bm{\gamma}}_{j,1},...,\tilde{\bm{\gamma}}_{j,j} \in \Gamma,\]
    Bernoulli random variables 
    \[\tilde{y}_{j,n} = \mathbb{I}_{\mathbb{R}_-^{d_{\text{c}}}}\left(\bm{h}(\bm{x})\right), \quad \bm{x}\sim \pi_{\tilde{\bm{\gamma}}_{j,n}},\]
    where $\pi_{\tilde{\bm{\gamma}}_{j,n}}$ denotes the unique stationary measure for the closed-loop system \eqref{eq:closedloop} using $\bm{\gamma}_{j,n}$, and $\tilde{\riskfun}_{j}$ denotes the posterior \gls{gp} conditioned on $\tilde{\mathcal{D}}_j$.. 
    \vspace{0.1cm}
\end{lem}

\vspace{0.1cm}
\begin{proof}
    Assume the contrary is true. Then, there exist scalars $\marginerror,\marginprobability>0$ such that, for every $J\in\mathbb{N}$, there exists a $j\geq J$ and a sequence of training inputs $\bar{\bm{\gamma}}_{{j},1},...,\bar{\bm{\gamma}}_{{j},j} \in \Gamma$, such that 
    \begin{align*}
        \Pr\left(\frac{1}{j} \sum \limits_{n=1}^j \vert \riskfun(\bar{\bm{\gamma}}_{j,n}) - \bar{\riskfun}_{{j}}(\bar{\bm{\gamma}}_{j,n})\vert > \marginerror \ \Bigg\vert \ \bar{\mathcal{D}}_j\right) \geq \marginprobability,
    \end{align*}
    where
        \begin{align}\bar{\mathcal{D}}_j &= \{\bar{\bm{\gamma}}_{j,n}, \bar{y}_{j,n} \}_{n=1,\ldots,j},\\ 
        \tilde{y}_{j,n} &= \mathbb{I}_{\mathbb{R}_-^{d_{\text{c}}}}\left(\bm{h}(\bm{x})\right), \quad \bm{x}\sim \pi_{\tilde{\bm{\gamma}}_{j,n}},\end{align}
        and $\bar{\riskfun}_{{j}}$ denotes the posterior \gls{gp} conditioned on $\bar{\mathcal{D}}_j$. Since \[\bar{\bm{\gamma}}_{{j},1},...,\bar{\bm{\gamma}}_{{j},j}\] is fixed for every $j$ and $j$ can be made arbitrarily large, this contradicts \Cref{lem:ghosal}.
\end{proof}

\Cref{lem:ghosal_extended} implies pointwise convergence of the \gls{gp} model with the number of data, provided that the data is sampled from the stationary distributions. We now show that \Cref{lem:stationarymeasure} can be leveraged to recover guarantees similar to those of \Cref{lem:ghosal_extended} for the \gls{gp} model trained on data collected over a finite horizon, provided that the horizon is sufficiently long.

\begin{lem}
\label{lem:recovered_bounds}
    Let \Cref{ass:compactGamma,ass:rkhsnorm,ass:stableloop,ass:decayofsequences} hold. Then, for every $\marginerror,\marginprobability > 0$, there exists a $J\in\mathbb{N}$, such that
    \begin{align*}
        \Pr\left(\frac{1}{j} \sum \limits_{n=1}^j \vert \riskfun(\bm{\gamma}_{n}) - \hat{\riskfun}_{{j}}(\bm{\gamma}_{n})\vert > \marginerror \ \Bigg\vert \ {\mathcal{D}}_j \right) < \marginprobability .
    \end{align*}
     holds for all data sets
     \[{\mathcal{D}}_j = \{{\bm{\gamma}}_{n}, {y}_{n} \}_{n=1,\ldots,j}\]
     with $j\geq J$, where the training inputs ${\bm{\gamma}}_{1},...,{\bm{\gamma}}_{j} \in \Gamma$ of ${\mathcal{D}}_j$ are identical to those of ${\mathcal{D}}_{j-1}$ up to the $j$-the training input, and Bernoulli random variables
    \[{y}_{n} = \mathbb{I}_{\mathbb{R}_-^{d_{\text{c}}}}\left(\bm{h}(\bm{x}_{n,\Tn})\right),\]
    where $\bm{x}_{n,\Tn}$ corresponds to the state after $\Tn$ steps under the closed-loop dynamics \eqref{eq:closedloop} given $\bm{\gamma}_n$ and arbitrary initial state $\bm{x}_{{n},0}\in\mathcal{X}$, and the number of steps satisfy
    \begin{align}
        \Tn \geq \frac{1}{-\log(\varphi)}\left( \log\left({\vartheta}V(\bm{x}_{n,0})\right) + j \log\left(2\right) \right),
    \end{align}
    with $\varphi$ and $\vartheta$ as in \Cref{lem:stationarymeasure}. Here $\hat{\riskfun}_{j} \in \Gamma$ denotes the posterior \gls{gp} conditioned on ${\mathcal{D}}_j = \{{\bm{\gamma}}_{n}, {y}_{n} \}_{n=1,\ldots,j}$.
\end{lem}
\begin{proof}
    Choose arbitrary $\marginerror>$ and $\marginprobability>$. Note that, since ${y}_{n}$ are Bernoulli random variables, we can write 
    \begin{align*}
        &\Pr\left(\frac{1}{j} \sum \limits_{n=1}^j \vert \riskfun(\bm{\gamma}_{n}) - \hat{\riskfun}_{{j}}(\bm{\gamma}_{n})\vert > \marginerror \right)  \\ 
        = & \sum_{m=1}^{j^2} P_m \ \mathbb{I}_{\mathbb{R}_+} \left(\frac{1}{j} \sum \limits_{n=1}^j \vert \riskfun(\bm{\gamma}_{n}) - \hat{\riskfun}_{n}(\bm{\gamma}_{n} \vert \bm{y}_m)\vert -\marginerror  \right),
    \end{align*}
    where $\hat{\riskfun}_{n}(\bm{\gamma}_{n} \vert \bm{y}_m)$ denotes the \gls{gp} model conditioned on the measurements $\bm{y}_m \in \{0,1\}^n$, the vectors $
        \bm{y}_1,\ldots,\bm{y}_{j^2}$ correspond to all possible different $j$-dimensional vectors containing zeros and ones, and 
    \begin{align*}
        P_m = \prod_{n=1}^j p^m_n, 
    \end{align*}
    with
    \begin{align*}
         p^m_n = \begin{cases}
\mathbb{E}\left[\mathbb{I}_{\mathbb{R}_-^{d_{\text{c}}}}\left(\bm{h}(\bm{x}_{n,\Tn})\right)\right]\quad &\text{if} \ {y}_{m,n} =1 \\
1-\mathbb{E}\left[\mathbb{I}_{\mathbb{R}_-^{d_{\text{c}}}}\left(\bm{h}(\bm{x}_{n,\Tn})\right)\right] &\text{if} \ {y}_{m,n} =0,
        \end{cases}
    \end{align*}
    where ${y}_{m,n}$ denotes the $n$-th entry of the vector $\bm{y}_m$. Consider now the data set $\tilde{\mathcal{D}}_j = \{{\bm{\gamma}}_{n}, \tilde{y}_{n} \}_{n=1,\ldots,j}$ obtained by sampling from the stationary distribution  of the closed-loop system given $\bm{\gamma}_{n}$, i.e.,
    \[\tilde{y}_{n} = \mathbb{I}_{\mathbb{R}_-^{d_{\text{c}}}}\left(\bm{h}(\bm{x})\right), \quad \bm{x}\sim \pi_{\bm{\gamma}_{n}}.\] 
    By applying the same procedure as above, we obtain
    \begin{align*}
        & \Bigg \vert \Pr\left(\frac{1}{j} \sum \limits_{n=1}^j \vert \riskfun(\bm{\gamma}_{n}) - \hat{\riskfun}_{{j}}(\bm{\gamma}_{n})\vert > \marginerror \ \Bigg\vert \ {\mathcal{D}}_j \right) \\
        &\quad - \Pr\left(\frac{1}{j} \sum \limits_{n=1}^j \vert \riskfun(\bm{\gamma}_{n}) - \hat{\riskfun}_{{j}}(\bm{\gamma}_{n})\vert > \marginerror \ \Bigg\vert \ \tilde{\mathcal{D}}_j \right)\Bigg \vert  \\
        = & \sum_{m=1}^{j^2}  \vert P_m-\tilde{P}_m \vert   \mathbb{I}_{\mathbb{R}_+} \left(\frac{1}{j} \sum \limits_{n=1}^j \vert \riskfun(\bm{\gamma}_{n}) - \hat{\riskfun}_{n}(\bm{\gamma}_{n} \vert \bm{y}_m)\vert -\marginerror\right) ,
    \end{align*}
    where
    \begin{align*}
        \tilde{P}_m = \prod_{n=1}^j \tilde{p}^m_n, 
    \end{align*}
    and
    \begin{align*}
         \tilde{p}^m_n = \begin{cases}
\mathbb{E}_{\pi_{\bm{\gamma}}}\left[\mathbb{I}_{\mathbb{R}_-^{d_{\text{c}}}}\left(\bm{h} (\bm{x})\right) \right]\quad &\text{if} \ {y}_{m,n} =1 \\
1-\mathbb{E}_{\pi_{\bm{\gamma}}}\left[\mathbb{I}_{\mathbb{R}_-^{d_{\text{c}}}}\left(\bm{h}(\bm{x})\right) \right] &\text{if} \ {y}_{m,n} =0. 
        \end{cases}
        \end{align*}
        From \Cref{lem:stationarymeasure}, it follows that
        \begin{align}
        \label{eq:bound_deltapmn}
            \vert \tilde{p}^m_n - {p}^m_n\vert \leq \vartheta V(\bm{x}_{n,0}) \varphi^{\Tn}.
        \end{align}
        By inserting \eqref{eq:stepslowerbound} into \eqref{eq:bound_deltapmn},
        we obtain $\vert \tilde{p}^m_n - {p}^m_n\vert \leq 2^{-j}$, 
        implying 
        \begin{align*} &\vert {P}_m -\tilde{P}_m \vert \leq \tilde{P}_m \max\left\{1- \left(1 - 2^{-j}\right)^j , \left(1 + 2^{-j}\right)^j-1 \right\} \\ =  & \tilde{P}_m \left(\left(1 + 2^{-j}\right)^j-1 \right) \leq \tilde{P}_m. \end{align*}
        By employing \Cref{lem:ghosal_extended}, we obtain, for $j$ sufficiently large
          \begin{align*}
           & \Pr\left(\frac{1}{j} \sum \limits_{n=1}^j \vert \riskfun(\bm{\gamma}_{n}) - \hat{\riskfun}_{{j}}(\bm{\gamma}_{n})\vert > \marginerror \ \Bigg\vert \ {\mathcal{D}}_j \right) \\ \leq &
          \Pr\left(\frac{1}{j} \sum \limits_{n=1}^j \vert \riskfun(\bm{\gamma}_{n}) - \hat{\riskfun}_{{j}}(\bm{\gamma}_{n})\vert > \marginerror \ \Bigg\vert \ \tilde{\mathcal{D}}_j \right) \\
        & + \Bigg \vert  \Pr\left(\frac{1}{j} \sum \limits_{n=1}^j \vert \riskfun(\bm{\gamma}_{n}) - \hat{\riskfun}_{{j}}(\bm{\gamma}_{n})\vert > \marginerror \ \Bigg\vert \ {\mathcal{D}}_j \right) \\
        &\quad - \Pr\left(\frac{1}{j} \sum \limits_{n=1}^j \vert \riskfun(\bm{\gamma}_{n}) - \hat{\riskfun}_{{j}}(\bm{\gamma}_{n})\vert > \marginerror \ \Bigg\vert \ \tilde{\mathcal{D}}_j \right) \Bigg \vert \\
        \leq & (1+1) \frac{\marginprobability}{2} = \marginprobability.
    \end{align*}
        
\end{proof}

\Cref{lem:recovered_bounds} states that if we wait long enough between constraint-tightening parameter updates and data collection, we obtain a model that faithfully captures the long-term probability of constraint satisfaction over a finite set of points in $\Gamma$. This allows us to prove our main theorem.

We now prove the first part of \Cref{theorem:mainresult}, which states that the final optimization step in \Cref{alg:onlineupdatealgorithm} is feasible with high probability.



\begin{lem}
\label{lemma:feasibility}
    Let \Cref{ass:compactGamma,ass:decayofsequences,ass:stableloop,ass:rkhsnorm} hold. Then, for any $\marginerror>0$ and $\marginprobability \in (0,1)$, there exists a $\Tfinal\in \mathbb{N}$, such that \eqref{eq:gammasearchgp_final} is feasible with probability at least $1-\marginprobability$.
\end{lem}
\begin{proof}
    Since the kernel $k$ is continuously differentiable and $\Gamma$ is compact, \Cref{ass:rkhsnorm} implies that the function $\riskfun$ is Lipschitz continuous on $\Gamma$. This means there exists a set of non-zero measure $\Gamma_{\text{c}}$, such that $H(\bm{\gamma})>1-\delta + \frac{\varepsilon_{\text{feas}}}{2}$ holds for all $\bm{\gamma} \in \Gamma_{\text{c}}$, where $\varepsilon_{\text{feas}}$ is as in \Cref{ass:compactGamma}.
Let $\bm{\gamma}_i$, $i=1,\ldots,\Tfinal$ denote the sequence of constraint-tightening parameters generated by \Cref{alg:onlineupdatealgorithm}.
Define the finite set
\begin{align*}
\Gamma_{\Tfinal,\text{c}} = \left\{\bm{\gamma}_m \ \vert \ m\leq \Tfinal, \ \bm{\gamma}_m \in \Gamma_{\text{c}}, \right\}
\end{align*}
and recall that whenever $i=n\crandom, \quad n\in \mathbb{N}$, the constraint-tightening parameter vector $\bm{\gamma}_i$ is sampled from a uniform distribution on $\Gamma$. Hence, by the law of large numbers, with probability one
\begin{align}
    \lim_{\Tfinal\rightarrow \infty} \frac{\vert \Gamma_{\Tfinal,\text{c}} \vert}{\Tfinal}  \geq \frac{\text{vol}(\Gamma_{\text{c}})}{\text{vol}(\Gamma)} > 0.
\end{align}
Hence, due to \Cref{ass:compactGamma}, for $\Tfinal$ large enough, we have
\begin{align*}
\Pr\left( \left\vert \frac{\vert \Gamma_{\Tfinal,\text{c}} \vert}{\Tfinal} - \frac{\text{vol}(\Gamma_{\text{c}})}{\text{vol}(\Gamma)} \right\vert \geq \frac{\text{vol}(\Gamma_{\text{c}})}{2\text{vol}(\Gamma)} \right) \leq \frac{\marginprobability}{2}.
\end{align*}
Note that, if \eqref{eq:gammasearchgp_final} is infeasible, then \[\hat{\riskfun}_{\Tfinal}(\bm{\gamma}) <1-\delta \leq \riskfun(\bm{\gamma}) - \frac{\varepsilon_{\text{feas}}}{2}\] holds for all $\bm{\gamma} \in \Gamma_{\text{c}}$, which in turn implies
\begin{align}
    \frac{1}{\Tfinal} \sum \limits_{\bm{\gamma}\in \Gamma_{\text{c}}} \vert \riskfun(\bm{\gamma}) - \hat{\riskfun}_{{\Tfinal}}(\bm{\gamma})\vert \geq \frac{\vert \Gamma_{\Tfinal,\text{c}} \vert \varepsilon_{\text{feas}}}{2 \Tfinal}.
\end{align}
By employing \Cref{lem:recovered_bounds} and $\Tfinal$ large enough, we can bound the probability that \eqref{eq:gammasearchgp_final} is infeasible as
\begin{align*}
        & \Pr\Bigg(  \hat{\riskfun}_{{\Tfinal}}(\bm{\gamma}) < 1 - \delta \quad \forall \bm{\gamma} \in \left\{\bm{\gamma}_1, \ldots, \bm{\gamma}_{\Tfinal} \right\} \Bigg) \\
        \leq & \Pr\Bigg(\frac{1}{\vert \Gamma_{\Tfinal,\text{c}} \vert} \sum \limits_{\bm{\gamma}\in \Gamma_{\Tfinal,\text{c}}} \vert \riskfun(\bm{\gamma}) - \hat{\riskfun}_{{\Tfinal}}(\bm{\gamma})\vert > \frac{\vert \Gamma_{\Tfinal,\text{c}} \vert \varepsilon_{\text{feas}}}{2 \Tfinal} \ \bigg\vert \ {\mathcal{D}}_j \Bigg)  \\
        = & \Pr\Bigg(\Bigg(\frac{1}{\vert \Gamma_{\Tfinal,\text{c}} \vert} \sum \limits_{\bm{\gamma}\in \Gamma_{\Tfinal,\text{c}}} \vert \riskfun(\bm{\gamma}) - \hat{\riskfun}_{{\Tfinal}}(\bm{\gamma})\vert > \frac{\vert \Gamma_{\Tfinal,\text{c}} \vert \varepsilon_{\text{feas}}}{2 \Tfinal} \Bigg) \\
        & \qquad \bigcup \Bigg( \bigg\vert \frac{\vert \Gamma_{\Tfinal,\text{c}} \vert}{\Tfinal} - \frac{\text{vol}(\Gamma_{\text{c}})}{\text{vol}(\Gamma)} \bigg\vert \geq \frac{\text{vol}(\Gamma_{\text{c}})}{2\text{vol}(\Gamma)}  \Bigg)\ \bigg\vert \ {\mathcal{D}}_j \Bigg) \\
        & +  \Pr\Bigg(\Bigg(\frac{1}{\vert \Gamma_{\Tfinal,\text{c}} \vert} \sum \limits_{\bm{\gamma}\in \Gamma_{\Tfinal,\text{c}}} \vert \riskfun(\bm{\gamma}) - \hat{\riskfun}_{{\Tfinal}}(\bm{\gamma})\vert > \frac{\vert \Gamma_{\Tfinal,\text{c}} \vert \varepsilon_{\text{feas}}}{2 \Tfinal} \Bigg) \\
        & \quad \qquad \bigcup \Bigg( \bigg\vert \frac{\vert \Gamma_{\Tfinal,\text{c}} \vert}{\Tfinal} - \frac{\text{vol}(\Gamma_{\text{c}})}{\text{vol}(\Gamma)} \bigg\vert  < \frac{\text{vol}(\Gamma_{\text{c}})}{2\text{vol}(\Gamma)}  \Bigg)\ \bigg\vert \ {\mathcal{D}}_j \Bigg) \\
        \leq & \frac{\marginprobability}{2}  +  \Pr\Bigg(\Bigg(\frac{1}{\vert \Gamma_{\Tfinal,\text{c}} \vert} \sum \limits_{\bm{\gamma}\in \Gamma_{\Tfinal,\text{c}}} \vert \riskfun(\bm{\gamma}) - \hat{\riskfun}_{{\Tfinal}}(\bm{\gamma})\vert >  \frac{\vert \Gamma_{\Tfinal,\text{c}} \vert \varepsilon_{\text{feas}}}{2 \Tfinal} \Bigg) \\
        & \quad \qquad \bigcup \Bigg( \bigg\vert \frac{\vert \Gamma_{\Tfinal,\text{c}} \vert}{\Tfinal} - \frac{\text{vol}(\Gamma_{\text{c}})}{\text{vol}(\Gamma)} \bigg\vert  < \frac{\text{vol}(\Gamma_{\text{c}})}{2\text{vol}(\Gamma)}  \Bigg)\ \bigg\vert \ {\mathcal{D}}_j \Bigg) \\ 
        \leq  &  \frac{\marginprobability}{2} +  \Pr\Bigg( \Bigg( \frac{\Tfinal}{\vert \Gamma_{\Tfinal,\text{c}} \vert}\frac{1}{\Tfinal} \sum \limits_{n=1}^{\Tfinal} \vert \riskfun(\bm{\gamma}_{n}) - \hat{\riskfun}_{{\Tfinal}}(\bm{\gamma}_{n})\vert > \frac{\vert \Gamma_{\Tfinal,\text{c}} \vert \varepsilon_{\text{feas}}}{2 \Tfinal}  \Bigg) \\
        & \quad \qquad \bigcup \Bigg( \bigg\vert \frac{\vert \Gamma_{\Tfinal,\text{c}} \vert}{\Tfinal} - \frac{\text{vol}(\Gamma_{\text{c}})}{\text{vol}(\Gamma)} \bigg\vert  < \frac{\text{vol}(\Gamma_{\text{c}})}{2\text{vol}(\Gamma)}  \Bigg)\ \bigg\vert \ {\mathcal{D}}_j \Bigg) \\ 
        \leq  &  \frac{\marginprobability}{2} \\  + & \Pr\Bigg(  \frac{1}{\Tfinal} \sum \limits_{n=1}^{\Tfinal} \vert \riskfun(\bm{\gamma}_{n}) - \hat{\riskfun}_{{\Tfinal}}(\bm{\gamma}_{n})\vert >  \frac{\varepsilon_{\text{feas}}}{8} \Bigg(\frac{\text{vol}(\Gamma_{\text{c}})}{\text{vol}(\Gamma)}\Bigg)^2 \ \bigg\vert \ {\mathcal{D}}_j  \Bigg) \\ 
        \leq & \marginprobability,
    \end{align*} 
where the second inequality is due to the union bound.
\end{proof}

\textit{Proof of \Cref{theorem:mainresult}:} Since \Cref{lemma:feasibility} establishes the feasibility of \eqref{eq:gammasearchgp_final} with arbitrarily high probability, to prove \Cref{theorem:mainresult} we only need to show that the corresponding solution satisfies the chance constraints up to an arbitrarily small margin. To this end, we show that
\begin{align}
   \vert \hat{\riskfun}_{{\Tfinal}}(\bm{\gamma}) -   {\riskfun}(\bm{\gamma})\vert < \marginerror \quad \forall \bm{\gamma} \in \left\{\bm{\gamma}_1, \ldots, \bm{\gamma}_{\Tfinal} \right\}
\end{align}
holds with arbitrarily high probability $1-\marginprobability$, which implies the desired result. 

Recall that \Cref{alg:onlineupdatealgorithm} collects $\Tcollect = \ccollect \Tfinal$ data points before updating $\bm{\gamma}_i$. Hence, if there exists a $\bm{\gamma} \in \left\{\bm{\gamma}_1, \ldots, \bm{\gamma}_{\Tfinal} \right\}$, such that $\vert \hat{\riskfun}_{{\Tfinal}}(\bm{\gamma}) -   {\riskfun}(\bm{\gamma})\vert \geq \marginerror$, then
\[\sum \limits_{n=1}^{\Tfinal} \vert \hat{\riskfun}_{{\Tfinal}}(\bm{\gamma}_n) -   {\riskfun}(\bm{\gamma}_n)\vert \geq {\Tcollect\marginerror} = \ccollect\Tfinal \marginerror \]
holds. By applying the union bound, we obtain, for $\Tfinal$ high enough,
\begin{align*}
        & \Pr\Bigg(\exists \bm{\gamma} \in \left\{\bm{\gamma}_1, \ldots, \bm{\gamma}_{\Tfinal} \right\}:\quad  \vert \hat{\riskfun}_{{\Tfinal}}(\bm{\gamma}) -   {\riskfun}(\bm{\gamma})\vert \geq \marginerror  \Bigg) \\
        \leq & \Pr\Bigg( \frac{1}{\Tfinal} \sum \limits_{n=1}^{\Tfinal} \vert \hat{\riskfun}_{{\Tfinal}}(\bm{\gamma}_n) -   {\riskfun}(\bm{\gamma}_n)\vert \geq {\ccollect\marginerror} \Bigg) 
        \leq \marginprobability.
    \end{align*} 
   \hfill \scalebox{0.86}{$\blacksquare$}

\bibliography{AllPhDReferences, Dissertation_bib}
\bibliographystyle{unsrt}

\begin{IEEEbiography}[{\includegraphics[width=1in,height=1.25in,clip,keepaspectratio]{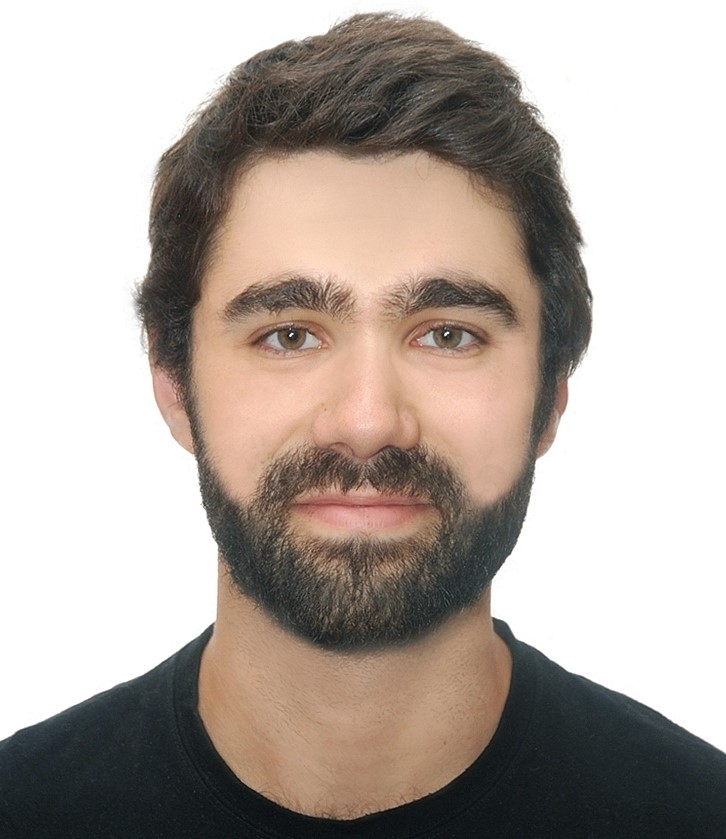}}]{Alexandre Capone}  received the B.S. and M.S. degrees in mechanical engineering from RWTH Aachen University, Aachen, Germany, in 2014 and 2016, respectively.

From 2017, he has been a PhD candidate with the Chair of Information-Oriented Control, at the Department of Electrical and Computer
Engineering of the Techical University of Munich, in Munich, Germany. He is currently visiting Amber lab under the supervision of Prof. Aaron Ames, at the California Institute of Technology. 
He was awarded the Springorium Commemorative coin for his studies at RWTH Aachen. His main research interests include \glspl{gp}, safe learning-based control, distributed control and event-triggered control, with applications to human-machine interaction and robotic systems.
\end{IEEEbiography}

\begin{IEEEbiography}[{\includegraphics[width=1in,height=1.25in,clip]{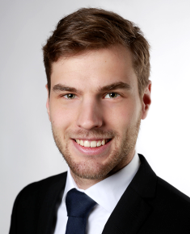}}]
{Tim Br\"udigam} received his
B.Sc. and M.Sc. degree in electrical engineering
from the Technical University of Munich (TUM),
Germany in 2014 and 2017, respectively. In 2022,
he obtained his Ph.D. degree from TUM. During
his studies, he was a research scholar at the California Institute of Technology in 2016. He joined
the Chair of Automatic Control Engineering at
TUM as a research associate in 2017 and was a
visiting researcher at the University of California,
Berkeley in 2021. His main research interest lies in advancing Model Predictive Control (MPC), especially stochastic MPC, with possible application
in automated driving.
\end{IEEEbiography}

\begin{IEEEbiography}[{\includegraphics[width=1in,height=1.25in,clip,keepaspectratio]{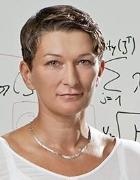}}]
{Sandra Hirche} received the
Diplom-Ingenieur degree in aeronautical engineering
from Technical University Berlin, Berlin, Germany,
in 2002, and the Doktor-Ingenieur degree in electrical engineering from Technical University Munich,
Munich, Germany, in 2005.

From 2005 to 2007, she was awarded a
Post-Doctoral Scholarship from the Japanese Society for the Promotion of Science, 
Fujita Laboratory, Tokyo Institute of Technology, Tokyo, Japan.
From 2008 to 2012, she was an Associate
Professor with the Technical University of Munich. She has been a TUM
Liesel Beckmann Distinguished Professor since 2013 and heads the Chair of
Information-Oriented Control with the Department of Electrical and Computer
Engineering, Technical University Munich. She has authored or coauthored
more than 150 articles in international journals, books, and refereed 
conferences. Her main research interests include cooperative, distributed and
networked control with applications in human–robot interaction, multirobot
systems, and general robotics.

Dr. Hirche has served on the Editorial Board of the IEEE Transactions on Control Network of Systems, 
the IEEE Transactions on Control Systems Technology, and IEEE Transactions on Haptics.
\end{IEEEbiography}

\end{document}